\documentclass[11pt,letterpaper]{article}
\usepackage[margin=1in]{geometry}
\usepackage{authblk}

\date{}
\usepackage{array}
\usepackage[font=small,labelfont=bf]{subfig}
\usepackage{textcomp}
\usepackage{url}
\usepackage{verbatim}
\usepackage{float}
\usepackage[utf8]{inputenc}
\usepackage[T1]{fontenc}
\usepackage{lmodern}
\usepackage{amsmath,amsfonts,amssymb}
\usepackage{mathtools}
\usepackage{amsthm}

\usepackage{nicefrac}
\usepackage{microtype}
\usepackage{bm}

\usepackage[hidelinks]{hyperref}
\usepackage[final]{graphicx}
\usepackage{xcolor}
\usepackage{makecell}

\usepackage{algorithm}
\usepackage{algorithmic}
\usepackage[numbers,square,comma,sort&compress]{natbib}
\usepackage{cleveref}
\usepackage{comment}
\usepackage{varwidth}

\usepackage{booktabs}
\usepackage[shortlabels]{enumitem}
\graphicspath{{figures/}}
\allowdisplaybreaks
\definecolor{OliveGreen}{rgb}{0,0.6,0}
\def\brho{\bm{\rho} }
\def\bA{\bm{A} }
\def\bI{\bm{I} }
\def\bC{\mathbb{C} }
\def\bW{\bm{W} }
\def\bU{\bm{U} }
\def\bu{\bm{u} }
\def\bUs{\bm{U}_{\star} }
\def\bus{\bm{u}_{\star} }
\def\bV{\bm{V} }
\def\E{\mathbb{E} }
\def\P{\mathbb{P} }
\def\O{\mathcal{O} }

\def\lV{\left\lVert }
\def\rV{\right\lVert }
\def\lv{\left\lvert}
\def\rv{\right\rvert}

\def\l{\left\langle}
\def\r{\right\rangle}
\def\fro{\mathrm{F}}
\def\Tr{\mathrm{Tr}}
\def\lambdas{\lambda^{\star}}
\def\bSigma{\bm{\Sigma}}
\def\brhos{\bm{\rho}_{\star}}

\Crefname{assumption}{Assumption}{Assumptions}
\Crefname{example}{Example}{Examples}
\Crefname{remark}{Remark}{Remarks}
\Crefname{claim}{Claim}{Claims}
\Crefname{lemma}{Lemma}{Lemmas}

\newtheorem{theorem}{Theorem}

\newtheorem{lemma}{Lemma}

\newtheorem{corollary}{Corollary}
\newtheorem{remark}{Remark}
\Crefname{assumption}{Assumption}{Assumptions}
\Crefname{example}{Example}{Examples}
\Crefname{remark}{Remark}{Remarks}
\Crefname{claim}{Claim}{Claims}

\def\BibTeX{{\rm B\kern-.05em{\sc i\kern-.025em b}\kern-.08em
    T\kern-.1667em\lower.7ex\hbox{E}\kern-.125emX}}
\begin{document}
\title{Online Quantum State Tomography via Stochastic Gradient Descent}
\author[1]{Jian-Feng Cai}
\author[2]{Yuling Jiao}
\author[2]{Yinan Li}
\author[3]{Xiliang Lu}
\author[4]{Jerry Zhijian Yang}
\author[2,5]{Juntao You}
\affil[1]{Department of Mathematics, Hong Kong University of Science and Technology, Hong Kong}
\affil[2]{School of Artificial Intelligence, Hubei Center for Applied Mathematics, and Hubei Key Laboratory of Computational Science, Wuhan University, Wuhan 430072, China}
\affil[3]{School of Mathematics and Statistics, Hubei Center for Applied Mathematics, and Hubei Key Laboratory of Computational Science, Wuhan University, Wuhan 430072, China}
\affil[4]{School of Mathematics and Statistics, Institute for Math \& AI, Hubei Center for Applied Mathematics, and Hubei Key Laboratory of Computational Science, Wuhan University, Wuhan 430072, China}
\affil[5]{Institute for Advanced Study, Shenzhen University, Shenzhen, China}

\maketitle

\begingroup
\renewcommand{\thefootnote}{}
\footnotetext[0]{
\begin{tabular}[t]{@{}l@{\qquad}l@{}}
J.-F. Cai: \texttt{jfcai@ust.hk} & Y. Jiao: \texttt{yulingjiaomath@whu.edu.cn} \\
Y. Li: \texttt{Yinan.Li@whu.edu.cn} & X. Lu: \texttt{xllv.math@whu.edu.cn} \\
J. Z. Yang: \texttt{zjyang.math@whu.edu.cn} & J. You: \texttt{youjuntao@whu.edu.cn} (corresponding author)
\end{tabular}
}
\endgroup

\begin{abstract}
We study online quantum state tomography (QST), where the matrix representation of an unknown quantum state is reconstructed by sequentially measuring the unknown states with respect to a batch of observables and updating the state estimate using only the measurement outcomes and observable information from the current round. Motivated by recent advances in non-convex optimization algorithms for solving low-rank QST, we propose non-convex mini-batch stochastic gradient descent (SGD) algorithms to solve online QST, which leverage the low-rank structure of the unknown quantum state and are well-suited for practical applications. Our main technical contribution is a rigorous convergence analysis of these algorithms. With proper initialization, we demonstrate that the squared error contracts locally in conditional expectation up to a finite-shot noise floor, and the iterates converge geometrically to this floor with high probability. Our algorithms achieve nearly optimal sample complexity up to condition-number and logarithmic factors, while remaining highly memory-efficient. In particular, their time complexities are better than the state-of-the-art non-convex QST algorithms, in terms of the rank and the logarithm of the dimension of the unknown quantum state. 
\end{abstract}


\section{Introduction}
\subsection{Background}
Quantum state tomography (QST) asks to recover the matrix representation of an \emph{unknown} quantum state $\brhos\in\bC^{d\times d}$ using the measurement statistics $\{y_i\approx\Tr(\bA_i\brhos):~i=1,\dots,m\}$ of a set of observables $\bA_1,\dots, \bA_m$. Despite its wide applications in quantum theory, quantum computation, and quantum information processing, the running time of QST algorithms can be extremely slow, as the matrix dimension $d$ grows exponentially with the number of individual quantum systems (the number of qubits). Experimental implementation of QST algorithms has only been achieved for very few qubits~\cite{PhysRevLett.113.040503,Riofro2017}, showcasing its computational difficulty.

Understanding the computational complexity of QST has been a fruitful research line in computer science and physics. There are many different QST schemes whose feasibilities rely on different hardware requirements and optimization algorithms. Consider the number of copies of the unknown state $\brhos$ needed to recover the matrix representation of $\brhos$ (sample complexity). If we are allowed to perform joint (entangled) measurements on all the copies of $\brhos$,~\cite{7956181} and~\cite{10.1145/2897518.2897544} proved that $\Theta(rd)$ many samples are sufficient and necessary to find an approximation of $\brhos$ up to constant error (in trace distance or infidelity distance)~\footnote{We shall focus on constant error setting in the introduction to simplify the presentation.}, where $r$ is the rank of the unknown state $\brhos$ (see also~\cite{wrightthesis,Yuen2023improvedsample}). These results can be improved if state-preparation unitaries are provided~\cite{doi:10.1137/1.9781611977554.ch47}, or generalized to performing joint measurements on a subset of all the copies~\cite{10.1145/3618260.3649704}. Recently, (optimal) memory complexity analysis of these QST schemes has been proposed in~\cite{hu2024sampleoptimalmemoryefficient}.

A major drawback of these schemes is that they require storing all copies of $\brhos$ simultaneously and performing highly entangled measurements — both of which remain challenging given the current hardware development. To address this limitation, one can consider QST schemes that rely on single-copy (unentangled) measurements. In particular, if random rank-$1$ measurements are allowed on each copy of the unknown state $\brhos$,~\cite{KUENG201788} proposed QST algorithms based on low-rank matrix recovery (see~\cite[Sec. 5.1]{wrightthesis} for another similar algorithm based on an empirical averaging technique). In this setting, $\Theta(r^2d)$ many copies of the unknown state are necessary and sufficient, even if the measurements are chosen adaptively~\cite{KUENG201788,7956181,wrightthesis,Guta_2020,lowe2022lowerboundslearningquantum,10353129,Flammia2024quantumchisquared,franca_et_al:LIPIcs.TQC.2021.7}. This indicates that entangled measurements are strictly powerful than unentangled measurements (even in the non-adaptive vs.~adaptive setting) in terms of sample complexity. 

Nevertheless, implementing random rank-$1$ measurements (when using certain techniques from matrix recovery algorithms), or implementing a sufficiently accurate approximate $4$-design~\cite{KUENG201788} instead, are still inefficient in practice due to the hardware restrictions. {Easily implementable measurements, such as the measurement with respect to (local) Pauli observables, are preferable for practical QST schemes.}~\footnote{For a Pauli observable $\bW$, which is the tensor product of $n$ Pauli operators on an $n$-qubit system, the corresponding measurement is the 2-outcome POVM $\{P_1=\frac{I+\bW}{2},P_{-1}=\frac{I-\bW}{2}\}$. The expected measurement outcome of this POVM on $\brhos$ is $\Tr(\bW\brhos)$.} In fact, single-copy local Pauli observables can produce enough information to reconstruct the unknown quantum states~\cite{PhysRevLett.105.150401,liu2011universal,flammia2011direct,Flammia_2012,10.1214/15-AOS1382,yu2020sampleefficienttomographypauli}. The most famous QST algorithms with Pauli observables are those based on \emph{compressed sensing}~\cite{1614066}, which have been implemented in practice~\cite{Riofro2017}.

Despite significant efforts to understand the sample complexity of QST, the time complexity of QST algorithms has received comparatively less attention. These algorithms typically involve processing large matrices through computationally expensive operations, such as matrix inversion, eigen-decomposition, or singular value decomposition—each requiring cubic time complexity, making them slow in practice. To mitigate these challenges, one can leverage the linear algebraic structure of quantum states. Many practical instances of QST involve recovering \emph{low-rank} quantum states (e.g., pure states). The low-rank structure has already been exploited in QST algorithms based on compressed sensing~\cite{PhysRevLett.105.150401,liu2011universal,Flammia_2012}, where certain convex optimization solvers have been utilized to provide time complexity estimations. 

More recently, \emph{non-convex optimization} techniques have been applied to optimization problems over low-rank matrices, yielding improved performance~\cite{NIPS2015_39461a19,pmlr-v49-bhojanapalli16,7536166,pmlr-v48-tu16,pmlr-v70-ge17a,doi:10.1137/17M1150189}. In particular, compressed sensing QST can be analyzed and implemented using the \emph{Projected Factored Gradient Descent} (ProjFGD) algorithm~\cite{kyrillidis2018provable,photonics10020116,9810003}, where the main idea is to utilize the low-rank decomposition $\brhos=\bU\bU^\dagger$ and work with the parameter matrix $\bU\in\bC^{d\times r}$, assuming the unknown state $\brhos$ is of rank (at most) $r$.~\cite{photonics10020116} showed that MiFGD, a variant of ProjFGD, can already outperform QST algorithms based on convex optimization, even those based on deep neural networks~\cite{Gao2017,Torlai2018,10.21468/SciPostPhys.7.1.009}. Utilizing the more advanced \emph{Riemannian Gradient Descent} (RGD) algorithm~\cite{hsu2024quantum}, the time complexity can be further improved in terms of the condition number of $\brhos$. We note that if the unknown quantum state satisfies certain sparsity conditions, the optimal convergence rate of QST algorithms was discussed in~\cite{10.1214/15-AOS1382}.

\subsection{Main Results}
In this paper, we initiate the study of low-rank QST algorithms with provable guarantees through the lens of \emph{online nonconvex optimization}. The online setting has been extensively explored in quantum learning theory, where the goal is to predict properties of unknown quantum states~\cite{Aaronson_2019,chen2020practicaladaptivealgorithmsonline,yang2020revisiting,Chen2024adaptiveonline,tseng2024onlinelearningquantumstates,lumbreras2024learningpurequantumstates,Fawzi2024,meyer2025onlinelearningpurestates} instead of obtaining the full matrix representation of the unknown quantum states (see~\cite{Anshu2024} for a survey of quantum state learning). Many interesting classical learning objectives, such as shadow tomography~\cite{doi:10.1137/18M120275X} and classical shadow~\cite{Huang2020}, have been proposed and proven useful for many quantum learning tasks, such as predicting ground-state properties of gapped Hamiltonians~\cite{doi:10.1126/science.abk3333}. 

For reconstructing the full matrix representation, online QST algorithms have been studied in~\cite{Youssry_2019,Zhang2020,pmlr-v238-tsai24a,lin2021maximumlikelihoodquantumstatetomography,lin2022maximumlikelihoodquantumstatetomography,tsai2022fasterstochasticfirstordermethod}, their (sample and time) complexity analysis is incomplete or incomparable with the offline QST algorithms. Meanwhile, it is worth noting that the aforementioned methods analyze the reconstruction problem without exploiting the low-rank structure, leading to significant challenges in terms of sample complexity, memory complexity, and computational cost.

{We focus on the following online setting: We perform the corresponding measurement with respect to independently and randomly sampled observables $A_1,\dots,A_T$ on the unknown state $\brhos$ and estimate the mean value $y_i\approx\Tr(A_i\brhos)$ with respect to each observable. Instead of using $\{y_1,\dots,y_T\}$ to reconstruct $\brhos$ directly, we aim to design a classical algorithm that takes $y_1,\dots,y_T$ as inputs sequentially and outputs a sequence of quantum state estimations $\brho_1,\dots,\brho_T$, where $\brho_t$ is computed based only on $\rho_{t-1}$, $A_t$ and $y_t$. In particular, our aim to estimate the smallest $T$ such that $\brho_T$ produced by the algorithm is an $\epsilon$-approximation of $\brhos$ in the Frobenius norm.

Our main motivation to investigate such a setting is its potential advantages on the \emph{experimental side}: Note that many QST instances focus on verifying the outcomes of quantum computation or quantum communication tasks. Preparing different measurements on the unknown state requires additional time to reinstall the selected measurement setting. If we have a classical algorithm as described above which takes only the outcome of the current measurement to update the quantum state estimation, we can perform the classical computation simultaneously with preparing the next measurements. This can potentially lead to more time-efficient QST protocols, optimizing both the duration of experimental measurements and the computational overhead of classical optimization. }

Our main motivation to investigate such a setting is its potential advantages from both practical and experimental perspectives: Note that many QST instances focus on verifying the outcomes of quantum computation or quantum communication tasks. Preparing different measurements on the unknown state requires additional time to reconfigure the selected measurement setting. If we have a classical algorithm as described above which takes only the outcome of the current measurement to update the quantum state estimation, we can perform the classical computation simultaneously with preparing the next measurements. Moreover, such an online mechanism is naturally suited to quantum state tracking, where the underlying state may evolve over time and its estimate needs to be continuously updated from sequential measurement outcomes \cite{rambach2023efficient}. This can potentially lead to more time-efficient QST protocols, optimizing both the duration of experimental measurements and the computational overhead of classical optimization.

Following the recent progress on non-convex QST~\cite{kyrillidis2018provable,photonics10020116,9810003,hsu2024quantum}, we propose \emph{simple online optimization algorithms for solving low-rank QST using single-copy Pauli observables}. Our algorithms are based on the \emph{(mini-batch) stochastic gradient descent} (SGD) method, which is particularly advantageous for large-scale problems due to its efficiency and scalability~\cite{hu2009accelerated,bottou2010large,duchi2011adaptive,li2018statistical}. Non-convex SGD algorithms have been extensively studied in online estimation, with theoretical analysis demonstrating its convergence and effectiveness in various high dimensional tasks, including matrix completion~\cite{jin2016provable,recht2013parallel}, matrix factorization~\cite{gemulla2011large,de2015global}, and tensor decomposition~\cite{ge2015escaping}.

Our SGD algorithms for online QST sequentially estimate the quantum state based on measurement outcomes collected among $T$ rounds,  where only a small number $B$ of randomly selected {(local) Pauli observables $\{\bm{A}_{t,k}\}_{k=1}^B$} are measured at each round $t$. For the online data, we have access to the \emph{noisy} measurement outcomes
\begin{align}
    y_{t,k}&= \Tr(\bm{A}_{t,k} \brhos)+z_{t,k},~ k=1,\dots,B\label{eq: pauli outcomes}
\end{align}
at round $t$, where $z_{t,k}$ denotes the statistical noise of the $k$-th observable $\bA_{t,k}$. Since only a few measurements are used in each iteration ($B$ is small), the online algorithms benefit from lower per-iteration complexity, albeit at the expense of increased iteration counts. The main contribution of this paper is a thorough convergence analysis of the mini-batch SGD algorithm for online QST. In particular, we prove the following:

{
\begin{theorem}[Informal]\label{thm:informal}
Let $\brhos\in\bC^{d\times d}$ be an unknown rank-$r$ quantum state of an $n$-qubit quantum system ($d=2^n$) with condition number $\kappa\leq\sqrt{dr}$. Let $\mathbb{W}$ be the set of all local Pauli observables on the $n$-qubit quantum system. Let $B\le \min\{40\kappa^{2/3},d\}$. There exists an (online) algorithm that utilizes $B$ Pauli observables $\bA_{t,1},\dots,\bA_{t,B}$ sampled from $\mathbb{W}$ uniformly and independently within each round $t$ and their noisy measurement outcomes $y_{t,k}$ as in Eq.~(\ref{eq: pauli outcomes}), with $L\ge\frac{112d\log d}{\varepsilon_0^2}$ shots per setting during refinement. For $0<\epsilon<1$, if the initial error is at most $\frac{\lambdas_r}{3}$ and $0<\varepsilon_0\le\min\{\frac{\epsilon}{2},\frac{\lambdas_r}{2}\}$, the following hold:
\begin{itemize}
    \item The quantum state estimation $\brho_t$ is computed in $\O(Brd)$ floating-point operations (FLOPs) within each round $t$;
    \item It takes at most $T=\O(B^{-1}\kappa^{2}rd\log d\log\frac{1}{\epsilon})$ many rounds to output a state estimate $\brho_T$ satisfying $\lV \brho_T-\brhos\rV_{\fro}\leq\epsilon$ with probability at least $1-3T/d^{10}$;
    \item The number of Pauli observables needed is $\O(\kappa^2rd\log d\max\{r\log^{5} d,\log \frac{1}{\epsilon}\})$.
\end{itemize}
\end{theorem}  
}

{For pure state tomography ($r=1$), which was recently shown to be as hard as low-rank mixed state tomography~\cite{pelecanos2025mixedstatetomographyreduces} , we compare the performance of our online SGD algorithms with related offline convex and non-convex algorithms in Table~\ref{tab:comparisonpurestate}. For general low-rank state tomography, we compare the performance of our online SGD algorithms with related offline non-convex algorithms in Table~\ref{tab:comparison}. Our online SGD algorithms offer greater flexibility during the iterative process. Roughly speaking, all these optimization algorithms require firstly finding a good initialization, which can be achieved by performing a spectral method on the measurement outcomes of a batch of random Pauli measurements. This part dominates the total number of observables needed for the whole algorithm.~\footnote{See~\cite[Lemma 4]{kyrillidis2018provable} or \cite[Lemma 2 in supplementary material]{hsu2024quantum} for the general observable number bound $\O(\kappa^2r^2d\log^6 d)$ for nonconvex algorithms.} Starting from the good initialization, the offline algorithms require at least a batch of $m = \Omega(rd\log^6 d)$ random Pauli measurements within each iteration to ensure the so-called \emph{restricted isometry property} (RIP)~\cite{liu2011universal}, which is critical for their convergence analysis~\cite{kyrillidis2018provable,photonics10020116,hsu2024quantum}. In contrast, our convergence analysis does not rely on RIP, and the probabilistic linear convergence holds even for small batch sizes $B$. In fact, our analysis demonstrates that \emph{the mini-batch SGD achieves an iteration complexity that is $B$-times faster than standard SGD}, provided the batch size $B$ is not excessively large. Thus, our online algorithms require much lower per-iteration time complexity for computing the quantum state estimation updates, and the total time complexity (per-iteration complexity $\times$ number of iterations) can be even better than the other non-convex algorithms (in terms of the rank $r$ and the logarithmic of the dimension $d$), providing a suitable initialization. An online initialization method is also introduced for the pure-state case, which improves the measurement complexity $m$ compared to existing methods, see \Cref{tab:comparisonpurestate}.}

\begin{table}[!t]
    \centering
    \footnotesize
    \caption{Complexity comparison of our SGD algorithms for online QST with other non-convex offline QST algorithms ($r=1$). Throughout, $\widetilde{\mathcal O}(\cdot)$ hides logarithmic factors.}\label{tab:comparisonpurestate} 
    \setlength{\tabcolsep}{1pt}
    \renewcommand{\arraystretch}{1.5}
    \begin{tabular}{|c|c|c|c|}
        \toprule
        Algorithms & \makecell[c]{Memory \cr complexity} & {Number of observables} $m$ & \makecell[c]{Copy\cr complexity}\\
        \midrule
        \makecell[c]{Convex Optimization~\cite{PhysRevLett.105.150401}} & $\O(d^2)$ & $\O(d\log^{6} d)$ & $\tilde{\O}(d^2/\epsilon^2)$\\
        \hline
        \makecell[c]{ProjFGD/\cr MIFGD~\cite{photonics10020116}} & $\O(d)$ & $\O(d\log^{6} d)$ & $-$\\
        \hline
        \makecell[c]{RGD~\cite{hsu2024quantum}} & $\O(d)$ & $\O(d\log^{6} d)$ & $\tilde{\O}(d^2/\epsilon^2)$\\
        \hline
        SGD (this work) & $\O(rd)$ & $\O\left(\kappa^2rd\log d\left(r\log^{5} d+\log \frac{1}{\epsilon}\right)\right)$ & $\tilde{\O}(\kappa^2 r^2d^2/\epsilon^2)$\\
        \bottomrule
    \end{tabular}
    \par\medskip
    \begin{tabular}{|c|c|c|c|}
        \toprule
        Algorithms & \makecell[c]{Per-iteration\cr complexity} & \makecell[c]{Number of iterations \cr  for $\epsilon$-solution} & \makecell[c]{{Time Complexity} \cr {with good initialization}}\\
        \midrule
        \makecell[c]{Convex Optimization~\cite{PhysRevLett.105.150401}} & $-$ & $-$ & $-$\\
        \hline
        \makecell[c]{ProjFGD/\cr MIFGD~\cite{photonics10020116}} & $\O(m d )$ & $\O(\log\frac{1}{\epsilon})$ & $\O(d^2\log^{6} d\log\frac{1}{\epsilon})$\\
        \hline
        \makecell[c]{RGD~\cite{hsu2024quantum}} & $\O(m d)$ & $\O(\log\frac{1}{\epsilon})$ & $\O(d^2\log^{6} d\log\frac{1}{\epsilon})$\\
        \hline
        SGD (this work) & $\O(rd)$ & $\O(\kappa^{2}rd\log d\log\frac{1}{\epsilon})$ & $\O(\kappa^{2}r^2d^2\log d\log\frac{1}{\epsilon})$\\
        \bottomrule
    \end{tabular}
\end{table}

\begin{table}[!t]
    \centering
    \footnotesize
    \caption{Complexity comparison of our SGD algorithms for online QST with other non-convex offline QST algorithms ($B\le\{40\kappa^{2/3},d\}$).}\label{tab:comparison} 
    \setlength{\tabcolsep}{1pt}
    \renewcommand{\arraystretch}{1.5}
    \begin{tabular}{|c|c|c|c|}
        \toprule
        Algorithms & \makecell[c]{Memory \cr complexity} & {Number of Observables} $m$ & \makecell[c]{Copy\cr complexity}\\
        \midrule
        \makecell[c]{Convex \cr Optimization~\cite{PhysRevLett.105.150401}} & $\O(d^2)$ & $\O(rd\log^{6} d)$ & $\tilde{\O}(r^2d^2/\epsilon^2)$\\
        \hline
        \makecell[c]{ProjFGD/\cr MIFGD~\cite{photonics10020116}} & $\O(rd)$ & $\O(\kappa^2r^2d\log^{6} d)$ & $-$\\
        \hline
        \makecell[c]{RGD~\cite{hsu2024quantum}} & $\O(rd)$ & $\O(\kappa^2r^2d\log^{6} d)$ & $\tilde{\O}(r^2d^2/\epsilon^2)$\\
        \hline
        SGD (this work) & $\O(rd)$ & $\O\left(\kappa^2rd\log d\left(r\log^{5} d+\log \frac{1}{\epsilon}\right)\right)$ & $\tilde{\O}(\kappa^2 r^2d^2/\epsilon^2)$\\
        \hline
        \makecell[c]{Mini-batch SGD \cr (this work)} & $\O(rd)$ & $\O\left(\kappa^2rd\log d\left(r\log^{5} d+\log \frac{1}{\epsilon}\right)\right)$ & $\tilde{\O}(\kappa^2 r^2d^2/\epsilon^2)$\\
        \bottomrule
    \end{tabular}
    \par\medskip
    \begin{tabular}{|c|c|c|c|}
        \toprule
        Algorithms & \makecell[c]{Per-iteration\cr complexity} & \makecell[c]{Number of iterations \cr  for $\epsilon$-solution} & \makecell[c]{{Time Complexity} \cr {with good initialization}}\\
        \midrule
        \makecell[c]{Convex \cr Optimization~\cite{PhysRevLett.105.150401}} & $-$ & $-$ & $-$\\
        \hline
        \makecell[c]{ProjFGD/\cr MIFGD~\cite{photonics10020116}} & $\O(m rd)$ & $\O(\kappa^{\alpha}\log\frac{1}{\epsilon})$, $\alpha\ge \frac{1}{2}$ & $\O(\kappa^{2+\alpha}r^3d^2\log^{6} d\log\frac{1}{\epsilon})$\\
        \hline
        \makecell[c]{RGD~\cite{hsu2024quantum}} & $\O(mr d)$ & $\O(\log\frac{1}{\kappa\epsilon})$ & $\O(\kappa^{2}r^3d^2\log^{6} d\log\frac{1}{\kappa\epsilon})$\\
        \hline
        SGD (this work) & $\O(rd)$ & $\O(\kappa^{2}rd\log d\log\frac{1}{\epsilon})$ & $\O(\kappa^{2}r^2d^2\log d\log\frac{1}{\epsilon})$\\
        \hline
        \makecell[c]{Mini-batch SGD \cr (this work)} & $\O(Br d )$ & $\O(\frac{1}{B}\kappa^{2}rd\log d\log\frac{1}{\epsilon})$ & $\O(\kappa^{2}r^2d^2\log d\log\frac{1}{\epsilon})$\\
        \bottomrule
    \end{tabular}
\end{table}

\subsection{Overview of the online SGD Algorithm and its analysis: \texorpdfstring{$B=1$}{B=1}}\label{sec:overview}
We briefly describe the online SGD algorithm for QST and the key ingredients for the convergence analysis, focusing on the setting of batch size $B=1$. More precisely, at each round $t$, the experimenter {performs the $2$-outcome POVM $\{P_1=\frac{I+\bA_{t}}{2}, P_{-1}=\frac{I-\bm{A}_t}{2}\}$ $L$ times on the unknown state $\brhos$ and compute the measurement outcome $y_t = \Tr(\bm{A}_{t} \brho_{\star}) + z_t$, where $\bm{A}_{t}=\bm{P}_{t,1} \otimes \bm{P}_{t,2} \otimes \cdots \otimes \bm{P}_{t,n}$} is obtained by sample each $\bm{P}_{t,j}$ from the Pauli matrices $\{I_{2},X,Y,Z\}$ uniformly at random, and $z_t$ denotes the statistical noise arising from finite measurement repetitions. Note that that $\bm{A}_{t}$ is chosen from the set of local Pauli observables $\mathbb{W}$ uniformly at random.

We utilize the standard squared loss $\hat{\ell}_t(\brho) = \frac{1}{4}\left( y_t - \Tr(\bm{A}_{t} \brho) \right)^2$ to quantify the quality of an estimation $\brho$ at each round $t$.  The \emph{key idea} is to update $\brho_t$ from $\brho_{t-1}$ using \emph{a single gradient descent step} to reduce $\hat{\ell}_t(\brho_t)$. Note that one can fully minimize $\hat{\ell}_t(\brho_t)$; while this will take additional iterations. To update $\brho_t$, we exploit the low-rank structure of the state $\brho$ and work with its \emph{parameterization} $\brho = \bU \bU^{\dag}$, where $\bU \in \mathbb{C}^{d \times r}$ denotes the parameter matrix of the rank-$r$ positive semidefinite matrix $\brho\in\bC^{d\times d}$. We then focus on the \emph{instantaneous loss of the parameterization}:
\begin{align*}
    \ell_t(\bU)&:= \hat{\ell}_t(\brho)= \frac{1}{4} \left( y_t - \Tr\left(\bm{A}_{t} \bU \bU^{\dag}\right) \right)^2.
\end{align*}
We update $\bU_t$ using the loss $\ell_t(\bU)$ and the previous estimate $\bU_{t-1}$ through the gradient descent update rule:
\begin{align}
    \bU_t&= \bU_{t-1} - \eta \nabla_{\bU}\ell_t(\bU_{t-1})\nonumber\\
    &= \bU_{t-1} - \eta\left[ \Tr(\bm{A}_{t} \bU_{t-1} \bU^{\dag}_{t-1}) - y_t \right] \bm{A}_{t} \bU_{t-1},\label{eq:OGD}
\end{align}
where $t=1,\dots,T$, $\eta>0$ is the learning rate (or step size) and the prediction at round $t$ is $\brho_t=\bU_{t}\bU_{t}^{\dag}$.

\begin{algorithm}[h]
\caption{\textbf{S}tochastic \textbf{G}radient \textbf{D}escent (SGD) for Online QST}
\label{algo:SGD}
\begin{algorithmic}[1]
\STATE \textbf{Input:} $T$, learning rate $\eta$, number of measurements $L$
\STATE Initialize $\bU_0$  
\FOR{$t = 1, \ldots, T$}
    \STATE {Choose $\bA_t$ from the set of local Pauli observables $\mathbb{W}$ uniformly at random}
    
    \STATE {Perform the 2-outcome measurement with respect to $\bm{A}_{t}$ $L$ times on the unknown state $\brhos$ to obtain a noisy  $y_t\approx\Tr(\bA_t\brhos)$}, 
    
    \STATE Update
    \begin{align*}
    \bU_{t}&= \bU_{t-1}- \eta\Big[ \Tr(\bm{A}_{t} \bU_{t-1} \bU^{\dag}_{t-1})-y_{t} \Big]\bm{A}_{t}\bU_{t-1}
\end{align*}
\ENDFOR
\STATE  \textbf{Output:} $\brho_T=\bU_T\bU_T^{\dag}$.
\end{algorithmic}
\end{algorithm}
The proposed algorithm is formalized in Algorithm \ref{algo:SGD}. {Within each iteration, the experimenter randomly chooses a Pauli observable $\bA_t\in\mathbb{W}$ and perform the $2$-outcome measurement $\{P_1=\frac{I+\bA_t}{2},P_{-1}=\frac{I-\bA_t}{2}\}$ $L$ times to obtain an estimate 
\[
y_t=\Tr(\bA_t\brhos)+z_t,
\]
where $z_t$ denote the statistical noise caused by performing only $L$ measurements on the unknown state. We then utlize $\bA_t$ and $y_t$ to update $\bU_t$.} Note that computing the gradient descent step is much more efficient than the other non-convex offline algorithms' updates: $\nabla_{\bU}\ell_t(\bU_{t-1})$ can be computed in $\O(rd)$ FLOPs, as $\bm{A}_{t}$ are sparse matrices. Our theoretical analysis, detailed in Section \ref{sec:theory} and Section  \ref{sec:theory p}, establishes that this online learning framework achieves local linear convergence in both expectation and high-probability regimes, conditioned on a relatively ``nice'' initialization of the unknown state. Specifically, $\bU_0$ is $\mathcal{O}(\lambda^{\star}_r)$-close to $\brhos$ in the Frobenius norm, i.e., $\lV \bU_0\bU_0^\dagger-\brhos\rV_\fro \le \frac{1}{3}\lambda_r^\star$, where $\lambdas_r$ is the smallest nonzero eigenvalue of $\brhos$.

The convergence analysis of our online SGD algorithms follows from a similar analysis of the gradient descent method as follows.
\paragraph{Road map of the analysis} Recall that a differentiable function  $f:\mathbb{R}^n\to \mathbb{R}$ is said to be \emph{$2L$-smooth} if we have
\begin{align}
    \lv f(\bm{x}+\bm{y})-f(\bm{x})-\l\nabla f(\bm{x}),\bm{y} \r \rv&\le L\lV \bm{y}\rV_2^2,\label{ineq:lip}
\end{align}  
$\forall \bm{x}, \bm{y} \in \mathbb{R}^n$, for some constant $L>0$. The gradient descent method for minimizing an objective function $ f(\bm{x})$ takes the following iterative form:
\begin{align*}
    \bm{x}_{t}&= \bm{x}_{t-1} - \eta \nabla f(\bm{x}_{t-1}),
\end{align*}
where $ \eta>0 $ is the learning rate. In particular,  for $f(\bm x)$ with \eqref{ineq:lip} holds, we have
\begin{align*}
    f(\bm{x}_{t})&\le f(\bm{x}_{t-1})-\eta \l \nabla f(\bm{x}_{t-1}),\nabla f(\bm{x}_{t-1})\r + L\eta^2\lV \nabla f(\bm{x}_{t-1})\rV_2^2\cr
    &\le f(\bm{x}_{t-1}) - \eta c \| \nabla f(\bm{x}_{t-1}) \|_2^2.
\end{align*}
where $c:=1 - L \eta>0$ provided $\eta<\frac{1}{L}$. Thus, for any differentiable $f$ whose minimal value is $0$, if \eqref{ineq:lip} holds, then provided $\eta<\frac{1}{L}$, it suffices to show that $\| \nabla f(\bm{x}_{t-1}) \|_2^2 \geq \mu f(\bm{x}_{t-1})$, $t>0$,  for some constant $\mu>0$ to have the linear convergence 
\begin{equation*}
    f(\bm{x}_{t}) \leq (1 - \eta \mu c) f(\bm{x}_{t-1}),~t=1,2,\cdots.
\end{equation*}

\paragraph{Key ingredients for the convergence of online QST} We utilize the above framework to analyze our SGD algorithm for online low-rank QST. 
Consider the \emph{expected loss in a simplified noiseless setting}, (i.e., $z_t=0$):\footnote{Here, the randomness induced by the measurement process is omitted for clarity of exposition. The full noisy setting is analyzed in~\Cref{sec:non_convex_mini_batch_sgd_for_online_qst}.}
\begin{align*}
    f(\bU)&:=4d\mathbb{E}[\ell_t(\bU)]=\lV\bU\bU^{\dag}-\brhos \rV_{\fro}^2 =\mathrm{dist}^2(\bU\bU^{\dag},\brhos),
\end{align*} 
where the expectation is taken over all random choices of measurements. We first illustrate that, if $\brho=\bU\bU^{\dag}$ is in an $\O(\lambdas_r)$-neighborhood of $\brhos$, there exists numerical constant $L>0$, such that for all perturbations satisfying  $\lV\bV \rV_{\fro}\le \O(\sqrt{\lambdas_r})$, it holds that 
\begin{align*}
    \lv f(\bU+\bV)-f(\bU)-\Re\l\nabla f(\bU),\bV \r \rv&\le L\lV \bV\rV_{\fro}^2.
\end{align*}
 
The local $2L$-smoothness of $f(\bU)$, together with the update rule~\eqref{eq:OGD}, implies the following local upper bound on the expectation value of the error metric at round $t$: {\small
\begin{align*}
    \E[f(\bU_{t})]&\le f(\bU_{t-1})-\eta \Re\big\langle \nabla f(\bU_{t-1}),\E[\nabla_{\bU}\ell_t(\bU_{t-1})]\big\rangle + L\eta^2\E\left[\lV \nabla_{\bU}\ell_t(\bU_{t-1})\rV_{\fro}^2\right].
\end{align*}}
It suffices to establish an appropriate lower bound for $\l \nabla f(\bU_{t-1}),\E[\nabla_{\bU}\ell_t(\bU_{t-1})]\r$ (the regularity term) and an appropriate upper bound for $\E[\lV\nabla_{\bU}\ell_t(\bU_{t-1})\rV_{\fro}^2]$ (the smoothness term). In the noiseless case, we establish the following regularity and smoothness conditions:
  \begin{align*}
    \Re\l \nabla f(\bU_{t-1}),\E[\nabla_{\bU}\ell_t(\bU_{t-1})]\r&\ge \Omega\left(\frac{\lambdas_r}{d}\right)f(\bU_{t-1}),\cr
    \E[\lV\nabla_{\bU}\ell_t(\bU_{t-1})\rV_{\fro}^2]&\le  \O\left(\frac{r}{d}\right)f(\bU_{t-1}).
\end{align*}
Then, conditioning on the current iterate, we have
\begin{align*}
    \E[f(\bU_{t})]&\le \left(1-\frac{\eta\lambdas_r}{2d}\right)f(\bU_{t-1}),
\end{align*}
provided that the learning rate is sufficiently small ($\eta\le \O(\frac{1}{\kappa r\lambdas_1})$). This translates to the following \emph{local contraction property}: If $\brho_{t-1}$ is in an $\O(\lambdas_r)$-neighborhood of $\brhos$, conditioning on the current iterate, we have
\begin{align*}
    \E[\lV\brho_{t}-\brhos \rV_{\fro}^2]&\le \left(1-\frac{\eta\lambdas_r}{2d}\right)\lV\brho_{t-1}-\brhos \rV_{\fro}^2,
\end{align*}
provided that $\eta \le \mathcal{O}\left(\frac{1}{\kappa r\lambdas_1}\right)$ and the measurement outcome is noiseless.\footnote{The noisy case is detailed later in~\Cref{sec:theory}.} 

Based on the standard Azuma-Bernstein concentration inequality, we further obtain a probabilistic convergence guarantee demonstrating that: If $\brho_0$ is in an $\O(\lambdas_r)$-neighborhood of $\brhos$, then for the prediction $\brho_t$ at round $t\le T$ it holds (in the noiseless case)
\begin{align*}
    \left\|\brho_{t} - \brho_{\star}\right\|_{\mathrm{F}}^2&\leq 2\left(1 - \frac{\eta\lambdas_r}{4d}\right)^t\left\|\brho_{0} - \brho_{\star}\right\|_{\mathrm{F}}^2,
\end{align*}
 with overwhelming probability, provided $\eta \leq \mathcal{O}\left(\frac{1}{\kappa r\lambdas_1\log d}\right)$. 
 
{\textbf{Remark on the choice of error metric.} We measure the reconstruction error in the Frobenius norm, which naturally arises from the least-squares measurement loss used in the algorithm. In particular, the expected noiseless loss is proportional to $\lV \bm{U}\bm{U}^\dagger-\brhos\rV_\fro^2$, enabling a tractable convergence analysis within the framework of nonconvex matrix factorization. Standard tomography metrics such as trace distance or infidelity are closely related to the Frobenius error. For instance, if $\operatorname{rank}(\brho),\operatorname{rank}(\brhos)\le r$, then $\lV \brho-\brhos\rV_{*} \le \sqrt{2r}\lV \brho-\brhos\rV_\fro$. Hence our guarantees directly imply bounds for trace distance up to a rank-dependent factor, while a direct analysis under these metrics would be technically more involved.}
\subsection{Online initialization via SGD}
From the above, we know that the proposed SGD algorithm achieves linear convergence, provided that the initial estimate $\brho_0$ lies within an $\O(\lambdas_r)$-neighborhood of the true state $\brhos$.  In this subsection, we present an online initialization method. For simplicity, we first consider the case that $\brhos$ is a pure state, i.e., $r=1$.  In this case, the density matrix can be decomposed as $\brhos=\bus\bus^{\dag}$, where $\bus\in \mathbb{C}^{d}$. Therefore, we aim to find a parameter vector that can reconstruct $\brhos$. For simplicity of presentation, we first consider $\|\rho_\star\|_2=1$. In this case, the leading eigenvector of $\rho_\star$ can be used directly as the parameter vector.\footnote{If $\|\rho_\star\|_2\neq 1$, one may estimate $(\lambda_1^\star,u_1)$ jointly and use $\sqrt{\lambda_1^\star}u_1$ as the parameter vector; the normalization is only for notational convenience.}  This leads to the solution to the following optimization problem:
\begin{equation*}
\max_{\bu} \bu^{\dag }(\brhos)\bu\quad\text{s.t.}\quad\lV\bu\rV_{2} = 1.
\end{equation*}
However, $\brhos$ is the unknown density matrix to recover.  Noticing that $\E[ d y_t \bm{A}_{t}]=\brhos$, the problem is equivalent to 
\begin{equation}\label{eq:PCA}
\min_{\bu} -\bu^{\dag }(\E[d y_t \bm{A}_{t}])\bu\quad\text{s.t.}\quad\lV\bu\rV_{2} = 1.
\end{equation}
A simple algorithm for solving \eqref{eq:PCA} is the projected gradient descent:
$$\bu_t=\mathcal{P}_{\mathcal{C}}\left(\bu_{t-1}+\eta_t \E[d y_t\bm{A}_{t}]\bu_{t-1}\right),$$
where $-\E[d y_t \bm{A}_{t}]\bu$ is a scaled gradient of the objective function $-\bu^{\dag }(\E[d y_t \bm{A}_{t}])\bu$, and $\mathcal{P}_{\mathcal{C}}$  denotes projection onto the set $\mathcal{C}:=\{\bu\in\mathbb{C}^{d}:\lV\bu\rV_{2} = 1\}$.  Nevertheless, we also do not have access to exact $\E[d y_t \bm{A}_{t}]$, so we replace $\E[d y_t \bm{A}_{t}]$ by its streaming random samples $\{ d y_t \bm{A}_{t}\}_{t\ge 1}$, which gives the update
$$\bu_t=\mathcal{P}_{\mathcal{C}}\left(\bu_{t-1}+\eta_t d y_t\bm{A}_{t}\bu_{t-1}\right),\quad t=1,2,\cdots.$$
Since $\bm{A}_{t}$ is randomly sampled from $\mathbb{W}$, the algorithm is a (projected) SGD algorithm, as detailed in \Cref{algo:iniSGD}.
\begin{algorithm}[h]
\caption{Online Initialization via SGD}
\label{algo:iniSGD}
\begin{algorithmic}[1]
\STATE \textbf{Input:} $T_0$, learning rate $\eta_t$, number of measurements $L$
\STATE Choose $\bu_0$ uniformly at random from the unit sphere  
\FOR{$t = 1, \ldots, T_0$}
    \STATE {Choose $\bA_t$ from the set of local Pauli observables $\mathbb{W}$ uniformly at random}
    
    \STATE {Perform the 2-outcome measurement with respect to $\bm{A}_{t}$ $L$ times on the unknown state $\brhos$ to obtain a noisy  $y_t\approx\Tr(\bA_t\brhos)$},

    \STATE Update
    $$\tilde{\bu}_{t} = \bu_{t-1}+ \eta_t d y_t\bm{A}_{t}\bu_{t-1},\quad \bu_t= \tilde{\bu}_{t}/\lVert \tilde{\bu}_{t}\rVert_{2}.$$
\ENDFOR
\STATE  \textbf{Output:} $\bu_{T_0}$ and $\brho_0=\bu_{T_0}\bu_{T_0}^{\dag}$.
\end{algorithmic}
\end{algorithm}
For the online initialization, we present \Cref{thm:onlineIni}, which guarantees that, starting from a random vector, \Cref{algo:iniSGD} returns a $\delta$-accurate initial estimate $\brho_0$ in at most $\O(\delta^{-2}d\log^2 d)$ iterations, with probability at least $\frac{3}{4}$—a probability that can be amplified to $1 - \frac{1}{d}$ as discussed later.
\begin{theorem}\label{thm:onlineIni}
Let $\delta\in(0,1]$ be any fixed constant. For $\lV\brhos\rV_2=1$, $r=1$, and $L\ge d$, letting $\eta_t=\frac{\log d}{160d\log^2d+t}$, we have the output of \Cref{algo:iniSGD} satisfies
\begin{align*}
    \lVert\brho_0-\brhos \rVert_{\fro}&\le \delta
\end{align*} 
with probability at least $\frac{3}{4}$, provided $T_0\ge C_0 \delta^{-2}d\log^2 d$ for some universal constant $C_0>0$.
\end{theorem}
\begin{proof}
    The proof is deferred to \Cref{sec:proofofini}.
\end{proof}  
The success probability $\frac{3}{4}$ can be boosted to $1-\frac{1}{d}$  through $\mathcal{O}(\log d)$ independent executions of \Cref{algo:iniSGD}, followed by computation of the geometric median over the resultant estimators. This aggregation process maintains computational efficiency, as geometric medians admit efficient approximate computation \cite{cohen2016geometric}. The following corollary demonstrates this procedure.
\begin{corollary}\label{col:onlineIni}
For $r=1$, $L\ge d$, and $\eta_t=\frac{\log d}{160d\log^2d+t}$, let  $\{\brho_{T_0,j}\}_{j=1}^{J}$ be the outputs of running $J$ copies of \Cref{algo:iniSGD}, and $\brho_{T_0}$ be the geometric median of the $\{\brho_{T_0,j}\}_{j=1}^{J}$, i.e.,
$$ \brho_{T_0}\in \arg\min_{\brho\in\mathbb{C}^{d\times d}}\sum_{j=1}^{J}\lV \brho-\brho_{T_0,j}\rV_{\fro}.$$
Then, it holds
\begin{align*}
    \lVert\brho_{T_0}-\brhos \rVert_{\fro}&\le \delta
\end{align*}
with probability at least $1-\frac{1}{d}$, provided $J\ge 72\log d$ and $T_0\ge 16C_0 \delta^{-2}d\log^2 d$.
\end{corollary}
\begin{proof}
    The proof is deferred to \Cref{sec:proofofinihp}.
\end{proof}  

In fact, our analysis indicates that: for pure state tomography, $\O(\delta^{-2}d\log^3 d)$ iterations of SGD described in \Cref{algo:iniSGD} with a random initial guess is sufficient to output an $\delta$-close estimation of  $\brhos$ with probability at least $1-\frac{1}{d}$. Nevertheless, to output some sufficiently accurate estimation of $\brhos$, e.g., to output an $\varepsilon$-estimation with $\varepsilon\ll \O(1) $, the \Cref{algo:iniSGD} requires  $\O(\varepsilon^{-2}d\log^3 d )$ iterations, while the two-stage algorithm requires only $\O(d\log^3 d \log\frac{1}{\varepsilon})$ iterations. Therefore, we use~\Cref{algo:iniSGD} as an initialization algorithm.

For the case of the underlying density matrix $\brhos$ is of rank $r$, it can be decomposed as $\brhos = \bUs\bUs^{\dag}$, where $\bUs \in \mathbb{C}^{d\times r}$. We consider finding the parameter $\bUs$ to reconstruct $\brhos$. Though $\bUs$ is not unique, it is natural that we can find the top-$r$ leading eigenvectors and eigenvalues of $\brhos$ to formulate $\bUs$. If $\brhos$ has $r$ distinct eigenvalues $\lambdas_1\ge\lambdas_2\ge\cdots\ge\lambdas_r$ which satisfies $\lambdas_i-\lambdas_{i+1}\ge \mathcal{O}\left(\frac{1}{r}\right)$ for $i=1,2,\cdots,r-1$, then it is possible to compute the top-$r$ leading eigenvectors of $\brhos$ sequentially via matrix deflation \cite{jain2016streaming}. Specifically, after the first leading eigenvector $\bm{u}_1$ and its eigenvalue $\lambdas_1$ are found using the online initialization method, their contribution is subtracted to form a deflated matrix $\brhos^{(1)} = \brhos - \lambdas_1 \bm{u}_1 \bm{u}_1^\dagger$, where the empirical streaming version of $\brhos^{(1)}$ is $d y_t^{(1)}\bA_t$ with $y_t^{(1)} =y_t-\widehat\lambda_1\widehat \bu_1^\dagger \bA_t\widehat \bu_1$. This suggests an online initialization by successive deflation. The guarantee in \Cref{thm:onlineIni} covers the pure-state case. The rank-$r$ guarantee below uses a spectral initializer. \textbf{}

\section{Preliminaries}\label{sec:notation}
Throughout the paper, we use regular lowercase letters for scalars (e.g., $d$), bold lowercase letters for vectors (e.g., $\bm{v}$), and bold capital letters for matrices (e.g., $\bm{A}$). The notation $[d]$ denotes the set $\{1, 2, \dots, d\}$ for any positive integer $d$. Given a vector $\bm{v}$, $\|\bm{v}\|_0$, $\|\bm{v}\|_2$, and $\|\bm{v}\|_\infty$ represent its $\ell_0$-, $\ell_2$-, and $\ell_\infty$-norms, respectively. For a matrix $\bm{M}$, $\lambda_i(\bm{M})$ denotes the $i$-th eigenvalue. For simplicity, we denote $\lambdas_i$ as the $i$-th eigenvalue of the unknown density matrix $\brhos$ and define $\kappa = \lambdas_1 / \lambdas_r$ as its condition number. $\Tr(\bm{M})$ denotes the trace of $\bm{M}$. We use $\lV\bm{M}\rV$ to represent the  spectral norm of $\bm{M}$ and $\lV\bm{M}\rV_F=\sqrt{\Tr(\bm{M}^\dagger\bm{M})}$ to represent the Frobenius norm of $\bm{M}$. $\lV\bm{M}\rV_{*}$ denotes the nuclear norm (or trace norm) of $\bm{M}$, given by $\lV\bm{M}\rV_{*}:=\Tr\left(\sqrt{\bm{M}^{\dag}\bm{M}}\right)$.

Consider a quantum system with $n$ qubits, where the associated density matrix $\brhos \in \mathbb{C}^{d\times d}$ is a positive semidefinite matrix of order $d = 2^n$ with trace unity.  Any such density matrix can be expressed as a linear combination of local Pauli observables $\mathbb{W}=\{\bA_i \mid i \in [d^2]\}$, where each Pauli observable is a tensor product of $n$ Pauli matrices $\bm{P}_1 \otimes \bm{P}_2 \otimes \cdots \otimes \bm{P}_n$, with each $\bm{P}_i$ selected from the set of $2 \times 2$ Pauli matrices $\{I_{2 \times 2}, X, Y, Z\}$. In fact, $\mathbb{W}$ form a complete orthogonal set in $\bC^{d\times d}$, satisfying the relation: $ \langle\bW_j,\bW_i\rangle=\Tr(\bW_i^\dagger \bW_j) = d \cdot \delta_{i,j} $ for all $i,j\in[d^2]$, where $\langle\cdot,\cdot\rangle$ denotes the Frobenius inner-product on $\bC^{d\times d}$ and $\delta_{i,j}$ is the Kronecker delta. This orthogonality allows any matrix $\bm{X} \in \mathbb{C}^{d \times d}$ to be uniquely represented as a linear combination of the Pauli observables: $ \bm{X} = \frac{1}{d} \sum_{i=1}^{d^2} \langle \bm{X}, \bW_i \rangle \bW_i. $ For the underlying density matrix $\brhos$, this expansion implies that each Pauli operator $\bW_i$ is associated with the coefficient $\frac{1}{d} \Tr(\bW_i \brhos)$ in the Pauli basis representation. {Let $(\Omega,\mathcal{F},\mathbb{P})$ be a probability space. A sequence of $\sigma$-algebras $\{\mathcal{F}_t\}_{t\ge 0}$ is called a \emph{filtration} if it is increasing, i.e.,
\[
\mathcal{F}_0 \subseteq \mathcal{F}_1 \subseteq \mathcal{F}_2 \subseteq \cdots
\subseteq \mathcal{F}.
\]
That says, $\mathcal{F}_t$ represents the information available up to time (or iteration) $t$.}

\section{Non-convex Mini-batch SGD for Online QST}\label{sec:non_convex_mini_batch_sgd_for_online_qst}
In this section, we present the mini-batch SGD algorithm for online low-rank QST in more detail. Recall that online low-rank QST aims to reconstruct an unknown rank-$r$ quantum state \(\brhos \in \mathbb{C}^{d \times d}\) of an \(n\)-qubit system (\(d=2^n\)) from sequentially performed measurement statistics. More precisely, within each round, we receive \(B\) measurement outcomes of the form
\begin{align}
    y_{t,k}&= \Tr(\bm{A}_{t,k} \brhos)+z_{t,k},~ k\in [B],\label{eq:minibatchmodel}
\end{align}
where each $\bm{A}_{t,k}$ is sampled uniformly at random from $\mathbb{W}$ and $z_{t,k}$ denotes the statistical noise incurred from the finite repetitions of the measurements. For a rank-\(r\) density matrix $\brho=\bU\bU^\dagger\in\bC^{d\times d}$, where $\bU\in\bC^{d\times r}$ is the parameter matrix,  the instantaneous loss $\ell_t(\bU)$ of $\bU$ and the squared loss $\hat{\ell}_t(\brho)$ of $\brho$ at the $t$-th round is defined as
 \begin{align*}
    \ell_t(\bU)&:=\frac{1}{4}\sum_{k=1}^B \left( y_{t,k} - \Tr(\bm{A}_{t,k} \bU \bU^{\dag}) \right)^2\cr
    &=~\frac{1}{4}\sum_{k=1}^B \left( y_{t,k} - \Tr(\bm{A}_{t,k} \brho) \right)^2
    =:\hat{\ell}_t(\brho).
\end{align*}
We shall work with the instantaneous loss $\ell_t(\bU)$ and update the parameter matrix $\bU$. Similar to the case of $B=1$, to update $\bU_t$ at round $t$, we use the previous estimation $\bU_{t-1}$ to reduce the loss $\ell_t(\bU_t)$. We follow the gradient descent update:
 \begin{align*}
    \bU_{t}&= \bU_{t-1}- \eta \nabla_{\bU}\ell_t(\bU_{t-1}),\quad t\in[T],
\end{align*}
where $\eta>0$ denotes the learning rate, which will be chosen explicitly from the analysis. By a direct computation, the gradient term $\nabla_{\bU}\ell_t(\bU_{t-1})$ at the $t$-th round is given by
\begin{align*}
    \nabla_{\bU}\ell_t(\bU_{t-1})&=\sum_{k=1}^B\left[ \Tr(\bm{A}_{t,k} \bU_{t-1} \bU_{t-1}^{\dag})-y_{t,k} \right]\bm{A}_{t,k}\bU_{t-1}.
\end{align*}
The proposed algorithm is formalized in \Cref{algo:minibatchGD}. The rest of this paper will focus on the convergence analysis of the algorithm: Assuming the initial guess of the state is sufficiently close to the target unknown state~\footnote{The initial guess can be computed by \Cref{algo:iniSGD} in the pure-state case or by an off-the-shelf spectral method. Spectral method yield an $\O(\lambdas_r)$-close approximation using $\O(\kappa^2r^2d\log^6 d)$ random Pauli measurements; see \cite[Lemma 4]{kyrillidis2018provable} or \cite[Lemma 2 in supplementary material]{hsu2024quantum} for details). More discussions on the initialization can be found in Appendix~\ref{appendix:numerical_simulation}, with numerical simulations.}, we shall provide bounds on the number of rounds $T$, the learning rate $\eta$ and the batch size $B$ such that the output state $\brho_T=\bU_T\bU_T^\dagger$ is $\epsilon$-close to the unknown state $\brhos$ with respect to the Frobenius norm. Note that the update can be computed in $\O(Brd\log d)$ FLOPs.
 \begin{algorithm}[htbp!]
\caption{Mini-batch SGD for Online QST}
\label{algo:minibatchGD}
\begin{algorithmic}[1]
\STATE \textbf{Input:} $T$, learning rate $\eta$, batch size $B$, number of measurement $L$ 
\STATE Initialize $\bU_0$  
\FOR{$t = 1, \ldots, T$}
    \FOR{$k = 1,\dots, B$}
    \STATE {Choose $\bA_{t,k}$ from the set of local Pauli observables $\mathbb{W}$ uniformly at random}
    
    \STATE {Perform the 2-outcome measurement with respect to $\bm{A}_{t,k}$ $L$ times on the unknown state $\brhos$ to obtain a noisy  $y_{t,k}\approx\Tr(\bA_{t,k}\brhos)$},
    \ENDFOR
    \STATE Update 
    \begin{align*}
    \bU_{t}&=~ \bU_{t-1} - \eta \sum_{k=1}^B\Big[ \Tr(\bm{A}_{t,k} \bU_{t-1} \bU_{t-1}^{\dag}) -y_{t,k} \Big]\bm{A}_{t,k}\bU_{t-1}
\end{align*}
\ENDFOR
\STATE \textbf{Output:} $\brho_T = \bU_T \bU_T^{\dagger}$.
\end{algorithmic}
\end{algorithm}

\subsection{Expectation convergence}\label{sec:theory}
We first define the local contraction region, a crucial subset of the parameter space where the algorithm exhibits desirable convergence behavior. Formally, this region is defined as
\begin{align*}
    \mathcal{E}(\brho_\star,\delta )&:=\left\{\bU: \lV \bU\bU^{\dag}-\brhos\rV_\fro\le \delta  \right\},
\end{align*} 
where $\delta\in [0, 1)$ is the diameter which will be chosen later. We shall consider the distance between the current estimate $\brho_t=\bU_t\bU_t^\dagger$ and the unknown quantum state $\brhos$. Define the learning error at $t$-th round of \Cref{algo:minibatchGD} as
\begin{align*}
    e_t&=\mathrm{dist}(\brho_t,\brhos):=\lV\brho_t-\brhos \rV_{\fro} =\lV \bU_t\bU_t^{\dag}-\brhos\rV_\fro.
\end{align*}

We first analyze the statistical error incurred by the finite repetitions of measurements. We demonstrate that the measurement error introduced by the shots in each round has a variance with zero mean and is well-bounded as long as the number of shots is sufficiently large.
\begin{lemma}\label{lemma:Eerror}
    At any round $t$, we have
     \begin{align*}
    \E[z_{t,k}\bA_{t,k}]&=\bm{0},~\forall k\in [B] .
\end{align*}
    Moreover, for $B\le d^2$ and $\varepsilon_0\in(0,1)$, provided {the number of measurements for each Pauli observable $L \ge 112\varepsilon_0^{-2} d\log d$}, it holds that
    \begin{align}
    \lv z_{t,k}\rv&\le \frac{\varepsilon_0}{\sqrt{d}},~\forall k\in [B]\label{ineq:boundzt}
\end{align}
    with probability at least $1-2d^{-10}$. 
\end{lemma}
\begin{proof}
The proof is deferred to \Cref{sec:proofoferror}.
\end{proof}

Now we present the local contraction property in expectation. Specifically, for iterates within the defined contraction region, the algorithm achieves a linear rate of error reduction per iteration in expectation, provided the learning rate is sufficiently small.  Let $\mathcal{F}_t$ be the filtration generated by the initial factor and all settings and shot outcomes through round $t$.
\begin{theorem}\label{thm:Econtraction}
Assume $B\le d$, $\kappa\le\sqrt{dr}$, $0<\varepsilon_0<1$, and $L\ge112\varepsilon_0^{-2}d\log d$. For $e_t\le\frac{\lambdas_r}{3}$, there exists a numerical constant $c_1>0$ such that
\begin{align*}
    \E[e_{t+1}^2|\mathcal{F}_t]&\le \left(1-\frac{\eta B\lambdas_r}{2d}\right)e_t^2+\frac{\eta B\lambdas_r\varepsilon_0^2}{8d},
\end{align*}
provided $\eta\le \frac{c_1}{\kappa r\lambdas_1}$ and $B\le 40\kappa^2$. Moreover, for $B\ge 40\kappa^2$, the same bound holds provided $\eta\le \frac{c_1\kappa}{Br\lambdas_1}$.
\end{theorem}
\begin{proof}
    The proof is deferred to \Cref{proofofEcontraction}
\end{proof}   
In the noiseless case, \Cref{thm:Econtraction} gives a local conditional contraction factor $1-\frac{\eta B\lambdas_r}{2d}$, provided the iterates remain in the stated neighborhood of $\brhos$.
\begin{remark}\label{remark:propconvergence}
In \Cref{thm:Econtraction}, we demonstrated that the mini-batch version of SGD achieves an iteration complexity that is $B$ times faster than standard SGD, provided the batch size $B$ is not excessively large. Notably, our convergence analysis holds for all values of $B\le d$. Specifically, for $\bU_t\in\mathcal{E}(\brhos,\frac{\lambdas_r}{3})$, we always have
\begin{align*}
    \E[e_{t+1}^2|\mathcal{F}_t]&\le \left(1-\frac{\eta B\lambdas_r}{2d}\right)e_t^2+\frac{\eta B\lambdas_r\varepsilon_0^2}{8d}
\end{align*} 
provided $\eta \le \O(\min\{\frac{\kappa}{B r\lambdas_1},\frac{1}{\kappa r\lambdas_1}\})$.
\end{remark}

In practice, the unknown quantum state $\rho^\star$ may not be exactly rank-$r$. For example, small depolarizing noise typically makes the state full-rank while preserving a rapidly decaying spectrum. To capture this scenario, we consider the approximate low-rank setting. Suppose $\mathrm{rank}(\brhos)\ge r$, and let $\brho_r$ denote the best rank-$r$ approximation of $\brhos$ in the Frobenius norm. By Weyl's theorem, $\rho_r$ is obtained by truncating the eigenvalue decomposition of $\brhos$. We write
\begin{align*}
    \brhos&= \brho_r + \bm{E},
\end{align*}
where $\bm{E}$ represents the spectral tail beyond rank $r$. The following corollary shows that the conditional one-step error relative to the best rank-$r$ approximation contracts locally up to a residual term depending on $\lV \bm{E}\rV_F$. For simplicity of exposition, we state the result for the case $B=1$.
\begin{corollary}\label{col:appLowrank}
Let $B=1$. Under the shot budget of \Cref{thm:Econtraction}, for $\bU_t\in\mathcal{E}(\brhos,\frac{\lambdas_r}{3})$, there exists numerical constant $c_1>0$ such that
\begin{align*}
    \E\left[\lV \bU_{t+1}\bU_{t+1}^{\dag}-\brho_r\rV_\fro^2\Bigg|\mathcal{F}_t\right]&\le \left(1-\frac{\eta\lambdas_r}{2d}\right)\lV \bU_t\bU_t^{\dag}-\brho_r\rV_\fro^2 +\frac{\eta\lambdas_r}{18d}\left(3\varepsilon_0^2+3\lV\bm{E}\rV_{\fro}^2+32\lV\bm{E}\rV_{\fro}\right)
\end{align*}
provided $\eta\le \frac{c_1}{\kappa r}$.
\end{corollary}
\begin{proof}
 The proof is deferred to \Cref{proofofappLowrank}.
\end{proof}

Corollary~2 shows that the proposed method remains stable when the true state is only approximately low-rank. In this setting, the algorithm should be interpreted as recovering the dominant rank-$r$ component $\brho_r$, while the spectral tail $E=\brho_\star-\brho_r$ contributes an additional perturbation term in the error bound. This is particularly relevant when some nonzero eigenvalues of $\brho_\star$ are very small: the corresponding eigendirections may be indistinguishable from finite-shot noise, and an effective-rank interpretation is more appropriate than exact recovery of all rank-$r$ components.

We note that the rank-$r$ truncation $\brho_r$ is generally a subnormalized positive semidefinite matrix rather than a density matrix, since it may satisfy $\operatorname{Tr}(\brho_r)<1$. To obtain a valid density matrix, one may project the final estimate onto
\[
\mathcal D_d
=
\{\brho\in\mathbb C^{d\times d}:\brho\succeq0,\ \operatorname{Tr}(\brho)=1\}.
\]
That is, for the output $\widehat\brho_T=\bU_T \bU_T^{\dagger}$, define
\[
\widehat\brho_T^{\rm dm}
=
\Pi_{\mathcal D_d}(\widehat\brho_T),
\]
where $\Pi_{\mathcal D_d}$ denotes the Frobenius projection onto $\mathcal D_d$. Since $\mathcal D_d$ is closed and convex and $\brho_\star\in\mathcal D_d$, this projection is nonexpansive, and hence
\[
\|\widehat\brho_T^{\rm dm}-\brho_\star\|_{\mathrm F}
\le
\|\widehat\brho_T-\brho_\star\|_{\mathrm F}
\le
\|\widehat\brho_T-\brho_r\|_{\mathrm F}
+
\|\brho_\star-\brho_r\|_{\mathrm F}.
\]
Thus this post-processing step produces a valid density matrix without increasing the Frobenius error to the true state.

\subsection{Probabilistic convergence}\label{sec:theory p}
We now convert the expectation convergence into a probabilistic convergence guarantee, showing that the algorithm achieves geometric convergence with high probability under practical conditions, including sufficient shots and appropriately chosen learning rates.
\begin{theorem}[Formal statement of \Cref{thm:informal}]\label{thm:propconvergence}
Assume  $L \ge 112\varepsilon_0^{-2}d\log d$, $0<\varepsilon_0\le\frac{\lambdas_r}{2}$, and $\kappa\le \sqrt{dr}$. There exists a numerical constant $c_2>0$ such that, for $\bU_0\in\mathcal{E}(\brhos,\frac{\lambdas_r}{3})$ and all $t\in [T]$, it holds
\begin{align*}
    e_t^2&\le 2\left(1-\frac{\eta B\lambdas_r}{4d}\right)^t e_0^2+\left[1-\left(1-\frac{\eta B\lambdas_r}{4d}\right)^t\right]\varepsilon_0^2,
\end{align*} 
with probability at least $1-\frac{3T}{d^{10}}$, provided $\eta \le \frac{c_2}{\kappa r\lambdas_1\log d}$ and $B\le \min\{40 \kappa^{2/3},d\}$. 
\end{theorem}
\begin{proof}
The proof is deferred to \Cref{subsec:proofofprobconvergence}
\end{proof}

This leads to the iteration complexity estimates for achieving a target accuracy $\epsilon$. 
\begin{corollary}\label{col:esterror}
Under the assumptions of \Cref{thm:propconvergence}, for $0<\epsilon<1$, we have    
\begin{align*}
    e_T^2&\le \O\left(\frac{\epsilon}{9}(\lambdas_r)^2+ \varepsilon_0^2\right)
\end{align*}
with probability at least $1-\frac{3T}{d^{10}}$, provided 
$$T\ge \frac{4}{c_2 B}\kappa^2 rd\log d\log\frac{1}{\epsilon}$$
and $\eta = \frac{c_2}{\kappa r\lambdas_1\log d}$, $B\le \min\{40 \kappa^{2/3},d\}$, $L \ge 112\varepsilon_0^{-2}d\log d$.
\end{corollary}
A standard spectral initializer \cite{kyrillidis2018provable,hsu2024quantum} attains the required entrance accuracy. For completeness, take the positive rank-$r$ truncation of the average of $m_0$ matrices $dy_i\bA_i$, with $L_0=dr$ shots per setting. Each centered summand has norm at most $d+1$ and variance norm at most $d\lV\brhos\rV_\fro^2+\frac{d^2}{L_0}\le2dr(\lambdas_1)^2$. Matrix Bernstein gives operator-norm error at most $\frac{\lambdas_r}{6\sqrt{2r}}$ with probability at least $1-d^{-10}$ for $m_0\ge C\kappa^2r^2d\log d$, using $\lambdas_1\ge\frac{1}{r}$. Truncation increases this error by at most a factor of two, and the difference has rank at most $2r$. Thus $e_0\le\frac{\lambdas_r}{3}$. At this choice of $m_0$, initialization costs $\widetilde{\O}(\kappa^2r^3d^2)$ copies, which is absorbed by the refinement cost when $\varepsilon\le\lambdas_r\le\frac{1}{r}$. Up to logarithmic factors, the proposed online initialization and non-convex SGD refinement require $\widetilde{\mathcal O}(d^2/\varepsilon^2)$ copies for pure-state tomography. For a general rank-$r$ state, combining a standard spectral initializer with Theorem~4 yields a copy complexity of $\widetilde{\mathcal O} \left( \frac{\kappa^2 r d^2}{\varepsilon^2} \right)$ for achieving $\varepsilon$-accuracy in the Frobenius norm, provided $\varepsilon\lesssim\lambda_r^\star$. The corresponding copy complexity for trace-norm accuracy $\varepsilon_*$ is $\widetilde{\mathcal O}(\kappa^2r^2d^2/\varepsilon_*^2)$.

\section{Numerical simulation and discussion on the initialization.} \label{appendix:numerical_simulation}

All numerical experiments were run on a laptop with an Intel(R) Core(TM) Ultra 5 125H processor and 32 GB of memory, under 64-bit Microsoft Windows 11, and implemented in MATLAB R2024b Update 3. We conduct four groups of numerical experiments. First, we revisit the original rank-one online SGD experiment to illustrate the effect of the mini-batch size $B$ under exact and finite-copy Pauli measurements. Second, we compare online SGD with offline low-rank QST solvers, including RGD and MiFGD, using the same Pauli data and reporting reconstruction error versus CPU time. Third, we study the convergence behavior of online SGD for approximately low-rank full-rank states. Finally, we examine the roles of online initialization and local SGD refinement, showing the transition from the initialization stage to the fast local convergence regime.

In the experiments, $T$ represent the number of rounds, and $B$ represents batch size (i.e., Pauli observables) per round, thus the total Pauli observables for SGD is $TB$.
\paragraph{Convergence Performance} We first conduct the rank-one experiments from the original numerical study to support the convergence performance of SGD for a $7$-qubit system with different batch sizes $B$. In these experiments, the $d\times d$ density matrix $\brhos$ is a randomly generated rank-$1$ positive semidefinite matrix. The initial guesses are all randomly generated according to $0.01\times\mathrm{randn}(d,r)$. The learning rate is $\eta=1/(4\kappa r)$ for $B\le 40$ and $\eta=50/(4B)$ for $B\ge 40$, in accordance with the theoretical scaling guidelines.

\begin{figure}[htbp!]
    \centering
    \includegraphics[width = 0.32\linewidth]{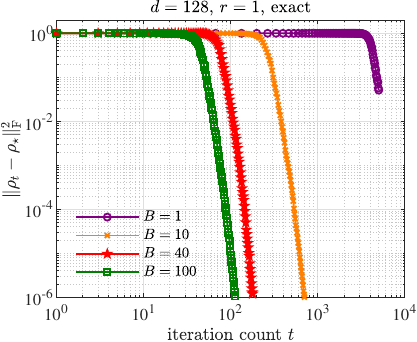}
    \includegraphics[width = 0.32\linewidth]{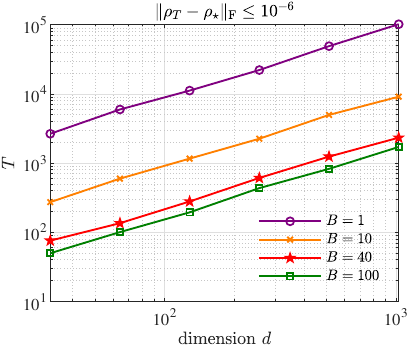}
    \includegraphics[width = 0.32\linewidth]{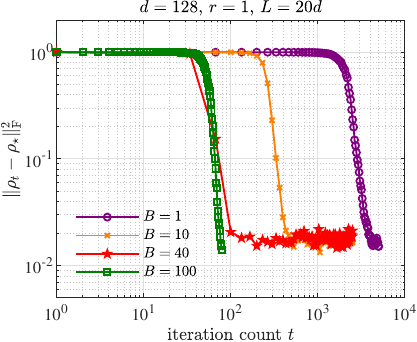}
    \caption{Batch-size dependence of online SGD for a rank-$1$ state with $n=7$ qubits. Left: exact Pauli measurement data. Middle: number of iterations (rounds) $T$ needed to achieve $\|\brho_T-\brhos\|_{\fro}\leq 10^{-6}$ with exact Pauli measurement data. Right: finite-shot setting using Pauli measurements with $L=20d$ shots for each $\bA_t$.}
    \label{fig:SGD_B}
\end{figure}

As shown in \Cref{fig:SGD_B}, the convergence process has two stages when the method starts from a small random factor. After a short transient phase, the exact-data curves exhibit geometric decay. The middle panel confirms the expected benefit of using a mini-batch: increasing $B$ reduces the required number of rounds $T$ until the batch size becomes large enough that the improvement saturates. In the finite-shot experiment, the same qualitative behavior is observed, while the finite number of shots introduces a nonzero statistical error floor.

\paragraph{Online and offline solvers under the same Pauli data}
We first compare CPU time under a fixed Pauli data set. For each rank $r\in\{1,3\}$, we draw one set of $m=3000$ Pauli observables and use the same outcomes for online SGD, RGD~\cite{hsu2024quantum}, and MiFGD~\cite{photonics10020116}. The online method uses mini-batches of size $B=20$ and cyclically reuses the fixed data set in this comparison. The two offline non-convex methods use the spectral initialization computed from the same $m$ outcomes. MiFGD is implemented  with momentum $\mu=0.5$ and step size $\eta=0.5/r$. We report both the exact-expectation setting and a finite-copy setting where each Pauli observable is measured with $L=20d=2560$ shots. 

\begin{figure}[htbp!]
    \centering
    \includegraphics[width = 0.46\linewidth]{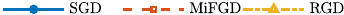}\\[0.3ex]
    \subfloat[exact, $r=1$]{\includegraphics[width = 0.32\linewidth]{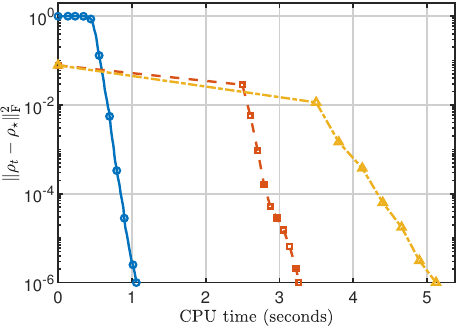}}
    \hfill
    \subfloat[exact, $r=3$]{\includegraphics[width = 0.32\linewidth]{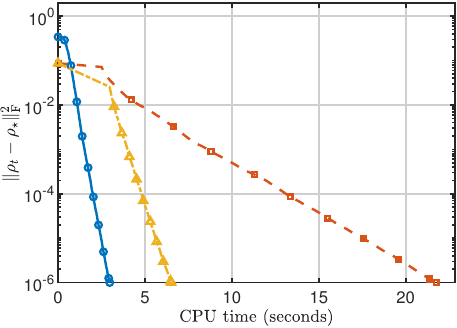}}
    \hfill
    \subfloat[$L=20d$, $r=3$]{\includegraphics[width = 0.32\linewidth]{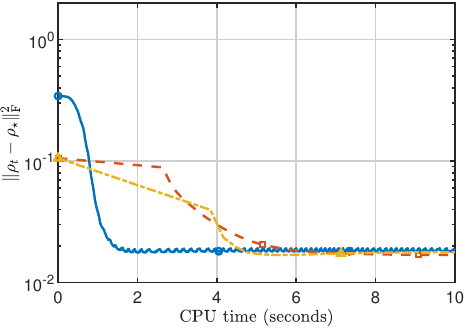}}
    \caption{CPU-time comparison on the same $m=3000$ Pauli observables for $n=7$ and $B=20$. Panels (a)--(b) use exact Pauli expectations; panel (c) uses empirical estimates from $L=20d=2560$ shots per Pauli observable.}
\label{fig:same_m_online_offline}
\end{figure}

\Cref{fig:same_m_online_offline} shows that all three non-convex methods converge on the exact-expectation data under the same CPU-time budget. In the finite-copy panel, the curves stabilize above the numerical floor because the empirical Pauli means contain shot noise; the final squared errors are about $1.8\times 10^{-2}$ for $r=3$. This comparison should be read together with the per-iteration costs in \Cref{tab:comparison}: the online update uses only a mini-batch of sparse Pauli actions, while each offline iteration processes all $m$ observables.

\paragraph{Approximately low-rank states}
We also test full-rank states with a dominant rank-$r$ component and a small isotropic tail,
\[
    \brhos=(1-\tau)\brho_r+\frac{\tau}{d}\bI,
\]
where $\brho_r$ is rank $r$. We use $r=3$, $\tau\in\{0.01,0.05\}$, $n=7$, and $B=20$, and compare these targets with exactly rank-$r$ states generated under the same protocol. For each approximately low-rank target we report both $\|\brho_t-\brhos\|_{\fro}^2$ and $\|\brho_t-\brho_r\|_{\fro}^2$.

\begin{figure}[htbp!]
    \centering
    \includegraphics[width = 0.94\linewidth]{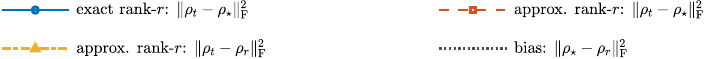}\\[0.3ex]
    \subfloat[$r=3$, $\tau=0.01$]{\includegraphics[width = 0.48\linewidth]{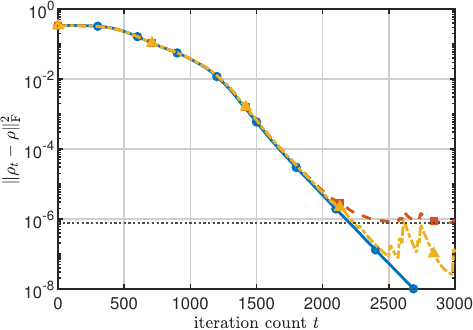}}
    \hfill
    \subfloat[$r=3$, $\tau=0.05$]{\includegraphics[width = 0.48\linewidth]{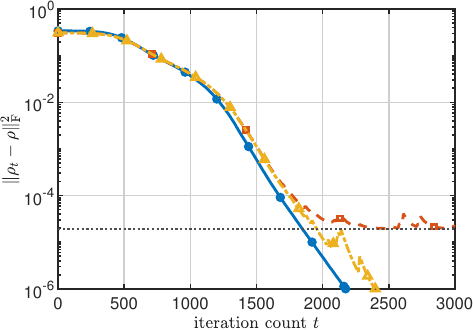}}
    \caption{Online SGD for exactly rank-$3$ states and full-rank approximately rank-$3$ states $\brhos=(1-\tau)\brho_r+\tau\bI/d$. The dotted curves mark the squared intrinsic biases $\|\brhos-\brho_r\|_{\fro}^2$. The histories are averaged over ten trials.}
    \label{fig:approx_low_rank}
\end{figure}

As shown in \Cref{fig:approx_low_rank}, the exactly low-rank curves continue to decrease to the numerical floor. For the full-rank targets, the final squared errors to $\brhos$ are close to the intrinsic approximation biases.  This behavior matches the perturbation term in the approximately low-rank analysis: the iteration tracks the dominant rank-$r$ component, while the isotropic tail gives a non-vanishing residual when the error is measured against the full state.

\paragraph{Initialization and online refinement}
Finally, we test the two stages of our method. The online initializer starts from a random unit vector for $r=1$ and from a random orthonormal $r$-frame for $r>1$, matching the random unit-sphere initialization in \Cref{algo:iniSGD}. We use a short online initialization window, with $T_0=300$, $\eta_0=5\times 10^{-3}$ for $r=1$, and $\eta_0=10^{-2}$ for $r=3$. The shaded region in \Cref{fig:online_initialization} marks this online initialization stage; after that point we switch to the local mini-batch SGD refinement.

\begin{figure}[htbp!]
    \centering
    \includegraphics[width = 0.50\linewidth]{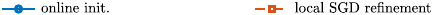}\\[0.3ex]
    \subfloat[$r=1$, $T_0=300$]{\includegraphics[width = 0.48\linewidth]{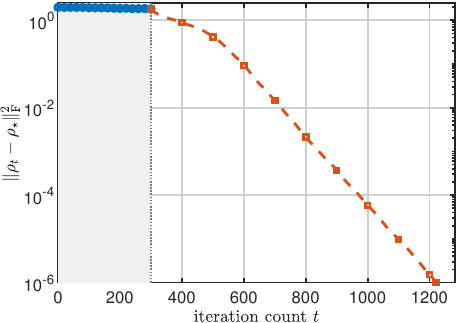}}
    \hfill
    \subfloat[$r=1$, first 600 iterations]{\includegraphics[width = 0.48\linewidth]{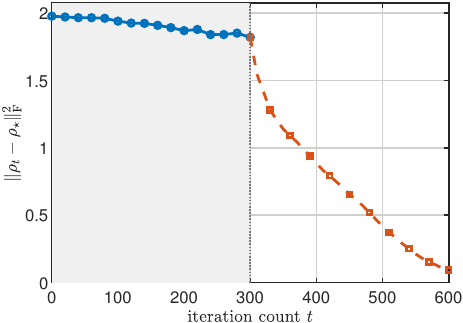}}\\[0.5ex]
    \subfloat[$r=3$, $T_0=300$]{\includegraphics[width = 0.48\linewidth]{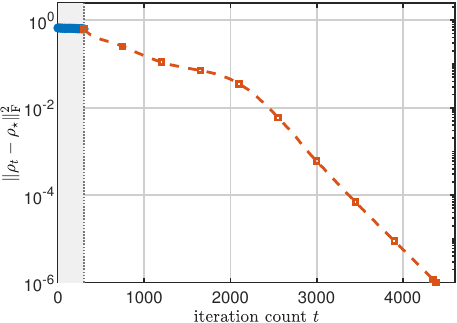}}
    \hfill
    \subfloat[$r=3$, first 600 iterations]{\includegraphics[width = 0.48\linewidth]{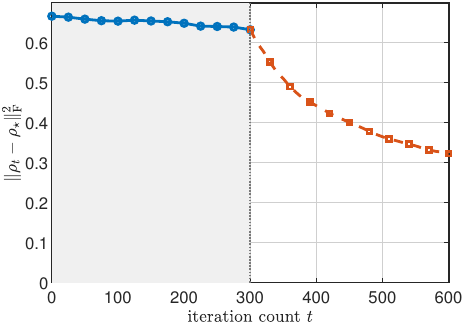}}
    \caption{Online initialization followed by local mini-batch SGD refinement for $n=7$, $B=20$, and $T_0=300$.  Panels (a) and (c) use a logarithmic error axis, while panels (b) and (d) show the first 600 iterations on ordinary linear axes.}
    \label{fig:online_initialization}
\end{figure}

The initialization stage decreases the error gradually, while the local refinement stage has a nearly straight trajectory on the semilogarithmic scale, indicating linear convergence once the iterate enters the local basin. 

\section{Proofs}
In this section, we provide comprehensive proofs of the main theoretical results. We begin with the proof of \Cref{lemma:Eerror} in \Cref{sec:proofoferror}, followed by the proof of \Cref{thm:onlineIni} in \Cref{sec:proofofini}. To establish \Cref{thm:Econtraction}, we first develop several key lemmas that establish the necessary local regularity and smoothness properties; these are presented in \Cref{subsec:keylemmas}. With these foundational results in place, we proceed to the full proof of \Cref{thm:Econtraction} in \Cref{proofofEcontraction}. Finally, \Cref{thm:propconvergence} is proven in \Cref{subsec:proofofprobconvergence} by applying the Azuma–Bernstein inequality.  We retain the physical scale in the initialization and geometric estimates. The rescaling used for the local SGD proofs is specified in the proof of \Cref{thm:Econtraction}.

\subsection{Proof of Lemma~\ref{lemma:Eerror}}\label{sec:proofoferror}
\begin{proof}
At round $t$, once each $\bA_{t,k} \in \{\bW_1, \bW_2, \dots, \bW_{d^2}\}$ is sampled to be $\bW_i$, the approximated coefficient $\Tr(\bW_i \brhos)$ is obtained from a 2-outcome measurement $\left\{\frac{I + \bW_i}{2}, \frac{I - \bW_i}{2}\right\}$ as in \eqref{eq:minibatchmodel}, with error $z_{t,k} = \hat{z}_{t,k,i}$. The outcome is a random variable $G_{t,k,i}$, where the subscript $i$ corresponds to the Pauli matrix $\bW_i$. Each instance of the random variable $G_{t,k,i}$ is denoted by $G_{t,k,i}^j$, with $j \in [L]$ referring to the $j$-th instance, and we perform $L$ measurements for each $\bW_i$. The instance $G_{t,k,i}^j = 1$ occurs with probability $\Tr\left(\frac{I + \bW_i}{2} \brhos\right)$, while $G_{t,k,i}^j = -1$ occurs with probability $\Tr\left(\frac{I - \bW_i}{2} \brhos\right)$. Thus, we obtain the measurement outcomes as

  $$\hat{y}_{t,k,i}=\frac{1}{L}\sum_{j=1}^{L}G_{t,k,i}^{j},~i\in [d^2].$$
 Therefore, 
  \begin{align}
    \E[z_{t,k}\bA_{t,k}]&=\frac{1}{d^2}\sum_{i=1}^{d^2}\E[\hat{z}_{t,k,i}|\bW_i] \bW_i\nonumber\\
    &= \frac{1}{d^2}\sum_{i=1}^{d^2}\E\left[\hat{y}_{t,k,i}-\Tr(\bW_i\brhos) |\bW_i\right] \bW_i\nonumber\\
    &= \frac{1}{d^2}\sum_{i=1}^{d^2}\sum_{j=1}^{L}\frac{1}{L}\E\left[G_{t,k,i}^{j}-\Tr(\bW_i\brhos)|\bW_i\right] \bW_i\nonumber\\
    &= \frac{1}{d^2}\sum_{i=1}^{d^2}\sum_{j=1}^{L}\frac{1}{L}\Bigl[\Tr\left(\frac{I+\bW_i}{2}\brhos\right) -\Tr\left(\frac{I-\bW_i}{2}\brhos\right)-\Tr(\bW_i\brhos)\Bigr]\bW_i\nonumber\\
    &= \bm{0}.\label{eq:EzA}
\end{align}
Now we prove the second statement. By definition we have
  $$\hat{z}_{t,k,i}=\sum_{j=1}^{L}\frac{1}{L}\left(G_{t,k,i}^{j}-\Tr(\bW_i\brhos)\right) :=\sum_{j=1}^{L}\frac{1}{L} \mathcal{Z}_{t,k,i}^{j}$$
where $\mathcal{Z}_{t,k,i}^{j}:=G_{t,k,i}^{j}-\Tr(\bW_i\brhos)$. Noticing $\E\left[\mathcal{Z}_{t,k,i}^{j}\right]=0$ and $\lv\mathcal{Z}_{t,k,i}^{j}\rv\le 2$, then by the Hoeffding's bound we have
\begin{align*}
    \P\left(\lv\sum_{j=1}^{L}\frac{1}{L} \mathcal{Z}_{t,k,i}^{j}\rv\ge \frac{\varepsilon_0}{\sqrt{d}} \right)&\le 2e^{-L \varepsilon_0^2/(8d)}.
\end{align*}
The conditional bound is uniform in the selected setting. Taking a union bound over the $B$ settings drawn at round $t$ gives
\begin{align*}
    \P\!\left(\max_{k\in[B]}|z_{t,k}|\ge\frac{\varepsilon_0}{\sqrt d}\right)&\le 2B e^{-L\varepsilon_0^2/(8d)}\le 2d^{-10},
\end{align*}
under the stated bounds on $B$ and $L$.
\end{proof}

\subsection{Proof of Theorem~\ref{thm:onlineIni}}\label{sec:proofofini}
\begin{proof}
 The proof is based on \cite[Theorem 3]{jain2016streaming}.   Noticing that $\{d y_t\bA_t\}_{t=1}^{T_0}$ is a sequence of matrices sampled independently from a distribution that satisfies 
    \begin{align*}
    \E[d y_t\bA_t]&= \frac{1}{d^2}\sum_{i=1}^{d^2}d \E[y_i|\bW_i]\bW_i\cr
    &=~ \frac{1}{d}\sum_{i=1}^{d^2}\Tr(\bW_i \brhos)\bW_i +\frac{1}{d}\sum_{i=1}^{d^2}\E\left[z_{i}|\bW_i\right]\bW_i=\brhos,
\end{align*}
    where the last equation follows from~\eqref{eq:EzA} and the fact that $\frac{1}{d}\sum_{i=1}^{d^2}\Tr(\bW_i \brhos)\bW_i=\brhos$. Moreover, as $|y_t|\le 1$ we have
     \begin{align*}
    \lV d y_t\bA_t-\brhos \rV_2&\le d\lV\bA_t\rV_2+\lV\brhos\rV_2 \le d+1,
\end{align*}
    and 
    \begin{align*}
    &\lV\E[(d y_t\bA_t-\brhos)(d y_t\bA_t-\brhos)^{\dag}]\rV_2 =~ \Bigl\lVert\E[d^2y_t^2\bA_t\bA_t^{\dag}] -\E[dy_t\bA_t\brhos^{\dag}+d\brhos y_t\bA_t^{\dag}]+\brhos\brhos^{\dag}\Bigr\rVert_2\cr
    &=~ \lV\sum_{i=1}^{d^2}\E\left[y_i^2|\bW_i\right]\bW_i\bW_i^{\dag}-\brhos\brhos^{\dag}\rV_2\cr
    &=~ \Biggl\lVert\sum_{i=1}^{d^2}\Bigl(\Tr(\bW_i \brhos)^2+2\Tr(\bW_i \brhos)\E\left[z_i|\bW_i\right] +\E\left[z_i^2|\bW_i\right]\Bigr)\bW_i\bW_i^{\dag}-\brhos\brhos^{\dag}\Biggr\rVert_2\cr
    &=~ \Biggl\lVert\sum_{i=1}^{d^2}\Tr(\bW_i \brhos)^2\bW_i\bW_i^{\dag} +\sum_{i=1}^{d^2}\E\left[z_i^2|\bW_i\right]\bW_i\bW_i^{\dag}-\brhos\brhos^{\dag}\Biggr\rVert_2\cr
    &=~ \lV\left(d+\frac{d^2-d}{L}\right)\bI-\brhos^2\rV_2=d+\frac{d^2-d}{L}\le 2d,
\end{align*}
where $\E[z_i^2|\bW_i]=\frac{1-\Tr(\bW_i\brhos)^2}{L}$ and $\sum_i\Tr(\bW_i\brhos)^2=d\lV\brhos\rV_\fro^2=d$. Since $L\ge d$, the variance is at most $2d$. In \cite[Theorem 3]{jain2016streaming}, we may use the bounds $d+1\le 4d\log d$ and $2d\le 4d$, together with spectral gap one, which gives the step size in \Cref{thm:onlineIni}. The matrix-product estimates extend to Hermitian matrices by replacing transposes with adjoints. For the random initialization, let $\bu_0=\frac{\bm{g}}{\lV\bm{g}\rV_2}$, where $\bm{g}$ is standard circular complex Gaussian and independent of the matrix product. For each fixed unit vector $\bm{v}$ and positive semidefinite matrix $\bm{H}$, we have $|\bm{v}^{\dag}\bm{g}|^2$ distributed as a mean-one exponential random variable and $\E[\bm{g}^{\dag}\bm{H}\bm{g}]=\Tr\bm{H}$. The exponential lower tail and Markov's inequality therefore give the same random-initialization estimate with constant probability. Applying \cite[Theorem 3]{jain2016streaming}, we obtain
\begin{equation*}
    1-\lv \bu_{T_0}^{\dag}\bu_{\star} \rv^2\le C\left(\frac{2d\log d}{T_0}+\left(\frac{320d\log^2d}{T_0}\right)^{2\log d}\right)\le \frac{\delta^2}{2}
\end{equation*}
with probability at least $\frac{3}{4}$, provided $T_0\ge C_0\delta^{-2}d\log^2d$ for a sufficiently large numerical constant $C_0>0$. Then, we have
\begin{align*}
    \lVert\brho_0-\brhos \rVert_{\fro}^2&=\lVert\brho_0 \rVert_{\fro}^2+\lVert\brhos \rVert_{\fro}^2-2\Re\l\bu_{\star}\bu_{\star}^{\dag},\bu_{T_0}\bu_{T_0}^{\dag}\r\cr
    &= 2\left(1-\l\bu_{T_0}^{\dag}\bu_{\star},\bu_{T_0}^{\dag}\bu_{\star}\r\right)\cr
    &= 2\left(1-\lv \bu_{T_0}^{\dag}\bu_{\star} \rv^2\right) \le \delta^2,
\end{align*}
which completes the proof.
\end{proof}

\subsection{Proof of Corollary~\ref{col:onlineIni}}\label{sec:proofofinihp}
\begin{proof}
Denote $$\mathcal{S}:=\{j:  \lVert\brho_{T_0,j} -\brhos \rVert_{\fro}> \frac{\delta}{4}\}.$$ We first show that w.h.p. we have $|\mathcal{S}|\le \frac{J}{3}$. \par Letting $T_0\ge C_0 \left(\frac{\delta}{4}\right)^{-2}d\log^2 d$. Then, for any fixed $j\in [J]$, by \Cref{thm:onlineIni}, it holds
$$\lVert\brho_{T_0,j}-\brhos \rVert_{\fro}\le \frac{\delta}{4}$$
with probability at least $\frac{3}{4}$. We define the independent random variables
\begin{align*}
    \xi_{j}&:=\mathbb{I}_{\{\lVert\brho_{T_0,j}-\brhos \rVert_{\fro}\le \frac{\delta}{4}\}},~j\in [J],
\end{align*}
where $\mathbb{I}_{\{\cdot\}}$ is the indicator function and $\mathbb{I}_{\{\mathbb{A}\}}=1$ if $\mathbb{A}$ is true, $\mathbb{I}_{\{\mathbb{A}\}}=0$ otherwise. Noticing that $\mathbb{E}[\xi_{j}]\ge\frac{3}{4}$, $j\in [J]$. By Hoeffding's inequality, we have
\begin{align*}
    \mathbb{P}\left(\frac{1}{J}\sum_{j=1}^{J}\xi_{j}-\E\left[\frac{1}{J}\sum_{j=1}^{J}\xi_{j}\right] \le -\varepsilon\right)&\le e^{-2J\varepsilon^2},
\end{align*}
which implies
\begin{align*}
    \mathbb{P}\left(\frac{1}{J}\sum_{j=1}^{J}\xi_{j}\le -\varepsilon+\frac{3}{4}\right)&\le~\mathbb{P}\left(\frac{1}{J}\sum_{j=1}^{J}\xi_{j}\le -\varepsilon+\E\left[\frac{1}{J}\sum_{j=1}^{J}\xi_{j}\right] \right)\cr
    &\le~ e^{-2J\varepsilon^2}.
\end{align*}
By letting $\varepsilon=\frac{1}{12}$ and $J\ge 72\log d$, we have
\begin{align*}
    \mathbb{P}\left(\sum_{j=1}^{J}\xi_{j}\ge \frac{2J}{3}\right)&\ge1- \frac{1}{d}.
\end{align*}
Therefore,
 \begin{align*}
    \mathbb{P}\left( |\mathcal{S}|\le \frac{J}{3}\right)&=\mathbb{P}\left(\sum_{j=1}^J \mathbb{I}_{\{\lVert\brho_{T_0,j}-\brhos \rVert_{\fro}\le \frac{\delta}{4}\}}\ge \frac{2J}{3}\right)\cr
    &= \mathbb{P}\left(\sum_{j=1}^{J}\xi_{j}\ge \frac{2J}{3}\right)\ge 1-\frac{1}{d}.
\end{align*}
Now we prove the statement by the properties of the geometric median. Let $\mathrm{vec}(\cdot)$ be the vectorization operator. Noticing that 
$$\lV \brho-\brho_{T_0,j}\rV_{\fro}=\lV \mathrm{vec}(\brho)-\mathrm{vec}(\brho_{T_0,j})\rV_{2}$$
for all $j$, we know $\mathrm{vec}(\brho_{T_0})$ is the geometric median of $\{\mathrm{vec}(\brho_{T_0,j})\}_{j=1}^{J}$. By optimality of the geometric median, the sum of the differences $\lV\brho_{T_0}-\brho_{T_0,j}\rV_\fro-\lV\brhos-\brho_{T_0,j}\rV_\fro$ is nonpositive. For $j\notin\mathcal{S}$, the triangle inequality bounds each difference below by $\lV\brho_{T_0}-\brhos\rV_\fro-2\max_{i\notin\mathcal{S}}\lV\brho_{T_0,i}-\brhos\rV_\fro$. For $j\in\mathcal{S}$, the difference is at least $-\lV\brho_{T_0}-\brhos\rV_\fro$. Summing these bounds, we have
\begin{align*}
    \lV\brho_{T_0}-\brhos\rV_\fro&\le\frac{2(J-|\mathcal{S}|)}{J-2|\mathcal{S}|}\max_{j\notin\mathcal{S}}\lV\brho_{T_0,j}-\brhos\rV_\fro\cr
    &\le 4\max_{j\notin\mathcal{S}}\lV\brho_{T_0,j}-\brhos\rV_\fro\le\delta.
\end{align*}
For subsequent factor updates, project the median onto the rank-one orthogonal projectors. The true state is feasible, so the triangle inequality increases the error by at most a factor of two. Applying the corollary at half the desired accuracy suffices. This completes the proof.
\end{proof}

\subsection{Key lemmas}\label{subsec:keylemmas}
Before we present the proof of \Cref{thm:Econtraction}, we first provide some technical lemmas.   Following the analysis illustrated in \Cref{sec:overview}, we examine the local smoothness of the squared distance function $\mathrm{dist}^2(\brho,\brhos)=\lV \bU\bU^{\dag}-\brhos\rV_\fro^2$, where $\brho=\bU\bU^{\dag}$.  
\begin{lemma}\label{lemma:Lip}
 The function $f(\bU)=\lV \bU\bU^{\dag}-\brhos\rV_\fro^2$ satisfies 
\begin{align*}
    \lv f(\bU+\bV)-f(\bU)-\Re\l\nabla f(\bU),\bV \r \rv&\le L_C\lV \bV\rV_{\fro}^2
\end{align*}
 
for all $\bU\in \mathcal{E}(\brho_\star,\delta\lambdas_r)$ and all $\lV\bV \rV_{\fro}\le C\sqrt{\lambdas_r}$, $C>0$. Here, $ L_C= (4\kappa+6\delta +4C\sqrt{\kappa+\delta}+C^2)\lambdas_r$.
\end{lemma} 
\begin{proof}
Through direct expansion, we have
 \begin{align*}
    &\lv f(\bU+\bV)-f(\bU)-\Re\l\nabla f(\bU),\bV \r\rv\cr
    &=~ \Bigl\lvert 2\Re\l \bU\bU^{\dag}-\brhos,\bV\bV^{\dag} \r+2\Re\l \bU\bV^{\dag}+\bV\bU^{\dag},\bV\bV^{\dag} \r +\lV \bU\bV^{\dag}+\bV\bU^{\dag} \rV_{\fro}^2+ \lV\bV\bV^{\dag} \rV_{\fro}^2\Bigr\rvert\cr
    &\le~ 2\lV\bU\bU^{\dag}-\brhos\rV_{\fro}\lV\bV\bV^{\dag} \rV_{\fro} +2\lV \bU\bV^{\dag}+\bV\bU^{\dag}\rV_{\fro}\lV \bV\bV^{\dag}\rV_{\fro} +\lV \bU\bV^{\dag}+\bV\bU^{\dag} \rV_{\fro}^2+ \lV\bV\bV^{\dag} \rV_{\fro}^2\cr
    &\le~ 2\lV\bU\bU^{\dag}-\brhos\rV_{\fro}\lV\bV \rV_{\fro}^2 +4\lV \bU\rV\lV \bV\rV_{\fro}^3 +4\lV \bU\rV^2\lV \bV\rV_{\fro}^2+\lV \bV\rV_{\fro}^4\cr
    &=~ \Bigl(2\lV\bU\bU^{\dag}-\brhos\rV_{\fro} +4\lV \bU\rV\lV \bV\rV_{\fro} +4\lV \bU\rV^2 +\lV \bV\rV_{\fro}^2\Bigr)\lV \bV\rV_{\fro}^2,
\end{align*}
where the first inequality follows from the Cauchy-Schwarz inequality, the second inequality follows from $\lV\bV\bV^{\dag} \rV_{\fro}\le \lV \bV\rV_{\fro}^2$. By Weyl's inequality, for all $\bU\in \mathcal{E}(\brho_\star,\delta\lambdas_r)$  we have
\begin{align}
    \lV \bU\rV^2&=\lV \bU\bU^{\dag}\rV\le \lambdas_1+\lV \bU\bU^{\dag}-\brhos\rV\nonumber\\
    &\le~ \lambdas_1+\lV \bU\bU^{\dag}-\brhos\rV_{\fro}\le \lambdas_1+ \delta\lambdas_r=(\kappa+\delta)\lambdas_r.\label{ineq:Uopup}
\end{align}
Therefore, 
\begin{align*}
    \lv f(\bU+\bV)-f(\bU)-\Re\l\nabla f(\bU),\bV \r\rv&\le~ \left(2\delta\lambdas_r +4C\sqrt{\kappa+\delta}\lambdas_r+4(\kappa+\delta)\lambdas_r +C^2\lambdas_r\right)\lV \bV\rV_{\fro}^2,
\end{align*}
which completes the proof.
\end{proof}

\begin{lemma}\label{lemma:boundEGDnorm}
    For $0\le\delta<1$ and all $\bU\in \mathcal{E}(\brho_\star,\delta\lambdas_r)$, we have
    \begin{align*}
    \Re\l \left( \bU \bU^{\dag}- \brhos\right)\bU\bU^{\dag},\left( \bU \bU^{\dag}- \brhos\right)\r&\ge~ \frac{1}{2}(1-\delta)^2\lambdas_r\lV\bU \bU^{\dag}- \brhos \rV_{\fro}^2.
\end{align*}
\end{lemma}
\begin{proof}
    Let $\brhos$ and $\bU\bU^{\dag}$ has the compact SVD $\brhos=\bV_{\star}\bSigma_{\star}\bV_{\star}^{\dag}$, $\bSigma_{\star}\in \mathbb{C}^{r\times r}$, $\bV_{\star}\in \mathbb{C}^{d\times r}$, and $\bU\bU^{\dag}=\bV\bSigma\bV^{\dag}$, $\bSigma\in \mathbb{C}^{r\times r}$, $\bV\in \mathbb{C}^{d\times r}$. Let $\theta$ be the largest principal angle between $\bV_\star$ and $\bV$. Since $(\bI-\bV\bV^{\dag})\bU=\bm0$, we have
\begin{align*}
    \sin\theta&=\lV (\bI-\bV\bV^{\dag})\bV_\star\rV_2\cr
    &\le \frac{1}{\lambdas_r}\lV(\bI-\bV\bV^{\dag})(\brhos-\bU\bU^{\dag})\bV_\star\rV_2 \le\delta .
\end{align*}
Therefore, by \cite[Theorem 2.1]{zhu2013angles}, we have
\begin{equation}\label{ineq:boundVsVsig}
    \sigma_{\min}^2(\bV_{\star}^{\dag}\bV)=\cos^2\theta=1-\sin^2\theta\ge 1-\delta^2.
\end{equation} 
Then, we have
\begin{align*}
    &\Re\l \left( \bU \bU^{\dag}- \brhos\right)\bU\bU^{\dag}, \bU \bU^{\dag}- \brhos\r =~ \Re\l \bV\bV^{\dag}\left( \bU \bU^{\dag}- \brhos\right)\bV\bSigma\bV^{\dag}, \bU \bU^{\dag}- \brhos\r\cr
    &\quad +\Re\bigl\langle (\bI-\bV\bV^{\dag})\left( \bU \bU^{\dag}- \brhos\right)\bV\bSigma, (\bI-\bV\bV^{\dag})(\bU \bU^{\dag}- \brhos)\bV\bigr\rangle\cr
    &=~ \Re\l \bV\bV^{\dag}\left( \bU \bU^{\dag}- \brhos\right)\bV\bSigma, \bV\bV^{\dag}(\bU \bU^{\dag}- \brhos)\bV\r +\Re\l (\bI-\bV\bV^{\dag}) \brhos\bV\bSigma, (\bI-\bV\bV^{\dag}) \brhos\bV\r\cr
    &\ge~ \lambda_r(\bSigma)\lV\bV\bV^{\dag}\left( \bU \bU^{\dag}- \brhos\right)\bV\bV^{\dag}\rV_{\fro}^2 +\lambda_r(\bSigma)\lV(\bI-\bV\bV^{\dag}) \bV_{\star}\bSigma_{\star}\bV_{\star}^{\dag}\bV \rV_{\fro}^2\cr
    &\ge~ \lambda_r(\bSigma)\lV\bV\bV^{\dag}\left( \bU \bU^{\dag}- \brhos\right)\bV\bV^{\dag}\rV_{\fro}^2 +\lambda_r(\bSigma)\sigma_{\min}^2(\bV_{\star}^{\dag}\bV)\lV(\bI-\bV\bV^{\dag}) \bV_{\star}\bSigma_{\star} \rV_{\fro}^2\cr
    &\ge~ \lambda_r(\bSigma)\lV\bV\bV^{\dag}\left( \bU \bU^{\dag}- \brhos\right)\bV\bV^{\dag}\rV_{\fro}^2 +(1-\delta)\lambda_r(\bSigma)\lV(\bI-\bV\bV^{\dag}) \bV_{\star}\bSigma_{\star} \bV_{\star}^{\dag}\rV_{\fro}^2
\end{align*} 
where the first inequality follows from the fact that: for positive semidefinite matrices $\bm{C}$, $\bm{D}$, we have  $\mathrm{tr}(\bm{C}\bm{D})\ge \lambda_{\min}(\bm{C})\mathrm{tr}(\bm{D})$. The last inequality follows from \eqref{ineq:boundVsVsig}. By noticing that 
\begin{align*}
    &\lV(\bI-\bV\bV^{\dag}) \bV_{\star}\bSigma_{\star} \bV_{\star}^{\dag}\rV_{\fro}^2 =~ \frac{1}{2}\left(\lV(\bI-\bV\bV^{\dag}) \brhos\rV_{\fro}^2+\lV\brhos(\bI-\bV\bV^{\dag}) \rV_{\fro}^2\right)\cr
    &=~ \frac{1}{2}\Bigl(\lV(\bI-\bV\bV^{\dag}) \brhos \bV\bV^{\dag}\rV_{\fro}^2 +\lV \bV\bV^{\dag}\brhos(\bI-\bV\bV^{\dag}) \rV_{\fro}^2 +2\lV (\bI-\bV\bV^{\dag})\brhos(\bI-\bV\bV^{\dag}) \rV_{\fro}^2 \Bigr),
\end{align*}
we have
\begin{align*}
    &\Re\l \left( \bU \bU^{\dag}- \brhos\right)\bU\bU^{\dag}, \bU \bU^{\dag}- \brhos\r\cr
    &\ge~ \frac{(1-\delta)\lambda_r(\bSigma)}{2}\Big(\lV\bV\bV^{\dag}\left( \bU \bU^{\dag}- \brhos\right)\bV\bV^{\dag}\rV_{\fro}^2 + \lV(\bI-\bV\bV^{\dag}) \left( \bU \bU^{\dag}- \brhos\right) \bV\bV^{\dag}\rV_{\fro}^2\cr
    &\quad +\lV \bV\bV^{\dag}\left( \bU \bU^{\dag}- \brhos\right)(\bI-\bV\bV^{\dag}) \rV_{\fro}^2 +\lV (\bI-\bV\bV^{\dag})\left( \bU \bU^{\dag}- \brhos\right)(\bI-\bV\bV^{\dag}) \rV_{\fro}^2\Big)\cr
    &\ge~ \frac{(1-\delta)^2\lambdas_r}{2}\lV\bU \bU^{\dag}- \brhos \rV_{\fro}^2
\end{align*} 
where the second inequality follows from $(\bI-\bV\bV^{\dag})\bU \bU^{\dag}=0$, the last inequality follows from the Weyl's inequality, which gives for all $\bU\in \mathcal{E}(\brho_\star,\delta\lambdas_r)$, $\delta \in (0,1)$, we have
\begin{align*}
    \lambda_r(\bSigma)&=\lambda_r(\bU\bU^{\dag})\ge \lambdas_r-\lV \bU\bU^{\dag}-\brhos\rV\cr
    &\ge \lambdas_r-\lV \bU\bU^{\dag}-\brhos\rV_{\fro}\ge (1-\delta)\lambdas_r.
\end{align*}
Keeping the bound $1-\delta^2$ from \eqref{ineq:boundVsVsig} in the preceding argument also gives
\begin{align*}
    \Re\l(\bU\bU^{\dag}-\brhos)\bU\bU^{\dag},\bU\bU^{\dag}-\brhos\r&\ge~\frac{(1-\delta)(1-\delta^2)\lambdas_r}{2}\lV\bU\bU^{\dag}-\brhos\rV_\fro^2.
\end{align*}
This completes the proof.
\end{proof}
We then obtain the following estimates, which include a perturbation bound as well as the regularity and smoothness conditions for the loss $\ell_t(\bU)$:
\begin{lemma}\label{lemma:boundSG}
Fix $\bU\in\mathcal{E}(\brhos,\delta\lambdas_r)$, $0\le\delta<1$, before drawing round $t$, with the shot budget of \Cref{thm:Econtraction}. The first bound below holds on \eqref{ineq:boundzt}; all the expectations condition only on the past:
    \begin{align}
    \lV\nabla_{\bU}\ell_t(\bU)\rV_{\fro}^2&\le~ 2 B^2\left(2r\lV\bU\bU^{\dag}-\brhos\rV_{\fro}^2+ \frac{\varepsilon_0^2}{d}\right)r(\kappa+\delta)\lambdas_r,\label{ineq:boundgd}\\
    \Re\l \nabla  f(\bU),\E[\nabla_{\bU}\ell_t(\bU)]\r&\ge~ \frac{2B}{d}  (1-\delta)^2\lambdas_r\lV\bU \bU^{\dag}- \brhos\rV_{\fro}^2,\label{ineq:boundEgd}\\
    \E[\lV\nabla_{\bU}\ell_t(\bU)\rV_{\fro}^2]&\le~ \frac{Br(\kappa+\delta)\lambdas_r}{d^2}\Bigl[3\max\{d,B\}\lV\bU\bU^{\dag}-\brhos\rV_{\fro}^2 +2d\varepsilon_0^2\Bigr],\label{ineq:Egrad2}\\
    \E[\lV\nabla_{\bU}\ell_t(\bU)\rV_{\fro}^4]&\le~ \frac{8B^4r^2(\kappa+\delta)^2(\lambdas_r)^2}{d^2}\Bigl(2rd\lV\bU\bU^{\dag}-\brhos\rV_\fro^4+\varepsilon_0^4\Bigr).\label{ineq:Egrad4}
\end{align}
\end{lemma}
\begin{proof}
We first prove \eqref{ineq:boundgd}. For all $\bU\in \mathcal{E}(\brho_\star,\delta\lambdas_r)$, we have
\begin{align}
    \lV\nabla_{\bU}\ell_t(\bU)\rV_{\fro}^2&=\lV\sum_{k=1}^B \left[ \Tr(\bm{A}_{t,k} \bU \bU^{\dag})-y_{t,k} \right]\bm{A}_{t,k}\bU\rV_{\fro}^2\nonumber\\
    &\le~ \left(\sum_{k=1}^B \lv \Tr(\bm{A}_{t,k} \bU \bU^{\dag})-y_{t,k} \rv\lV\bm{A}_{t,k}\bU\rV_{\fro}\right)^2\nonumber\\
    &\le~ \left(\sum_{k=1}^B \lv \Tr(\bm{A}_{t,k} \bU \bU^{\dag})-y_{t,k} \rv^2\right)\left(\sum_{k=1}^B\lV\bm{A}_{t,k}\bU\rV_{\fro}^2\right)\nonumber\\
    &\le~ 2\left( \sum_{k=1}^B \lv\l\bm{A}_{t,k}, \bU \bU^{\dag}-\brhos\r\rv^2+ \sum_{k=1}^B z_{t,k}^2 \right) \left(\sum_{k=1}^B\lV \bA_{t,k}\rV^2\lV\bU\bU^{\dag}\rV_{*}\right)\nonumber\\
    &\le~ 2\left(\sum_{k=1}^B\lV \bA_{t,k}\rV^2\lV\bU\bU^{\dag}-\brhos\rV_{*}^2+ \sum_{k=1}^B \frac{\varepsilon_0^2}{d}\right)B\lV\bU\bU^{\dag}\rV_{*}\nonumber\\
    &\le~ 2 B^2\left(2r\lV\bU\bU^{\dag}-\brhos\rV_{\fro}^2+ \frac{\varepsilon_0^2}{d}\right)r(\kappa+\delta)\lambdas_r,\label{ineq:proveboundgd}
\end{align}
where the second inequality follows from the Cauchy-Schwarz inequality, the third inequality follows from the inequality $\lV \bA_t\bU\rV_{\fro}^2=\l \bA_t\bA_t^{\dag},\bU\bU^{\dag} \r\le \lV\bA_t\rV^2\lV\bU\bU^{\dag}\rV_{*}$, the fourth inequality follows from the fact that $\bA_t$ is drawn from the set of standard Pauli matrices and thus $\lV\bA_t\rV\le 1$ (see \cite{liu2011universal}), the last inequality follows from~\eqref{ineq:Uopup} and the fact that 
\begin{align*}
    \lV\bU\bU^{\dag}-\brhos\rV_{*}&\le \sqrt{2r}\lV\bU\bU^{\dag}-\brhos\rV_{\fro},\cr
    \lV\bU\bU^{\dag}\rV_{*}&\le r\lV\bU\rV^2.
\end{align*}
Next, we prove \eqref{ineq:boundEgd}. For each $\bA_{t,k}$, $k\in [B]$ we have 
\begin{align*}
    \E\left[\left( \Tr(\bm{A}_{t,k} \bU \bU^{\dag})-\Tr(\bm{A}_{t,k}\brhos) \right)\bm{A}_{t,k}\bU\right]&=~ \frac{1}{d^2} \sum_{i=1}^{d^2} \left[ \Tr(\bW_i \bU \bU^{\dag})-\Tr(\bW_i \brhos)\right]\bW_i\bU.
\end{align*}
Thus,  by Lemma~\ref{lemma:Eerror} we have
  \begin{align}
    \E[\nabla_{\bU}\ell_t(\bU)]&=\frac{B}{d^2} \sum_{i=1}^{d^2} \left[ \Tr(\bW_i \bU \bU^{\dag})-\Tr(\bW_i \brhos)\right]\bW_i\bU - \sum_{k=1}^{B} \E[ z_{t,k} \bm{A}_{t,k}\bU]\nonumber\\
    &=~ \frac{B}{d^2}\sum_{i=1}^{d^2}\l \bU \bU^{\dag}- \brhos,\bW_i\r \bW_i\bU -\sum_{k=1}^{B} \E[ z_{t,k} \bm{A}_{t,k}]\bU\nonumber\\
    &=~ \frac{B}{d}\left( \bU \bU^{\dag}- \brhos\right)\bU,\label{eq:Egrad}
\end{align}
and thus by \Cref{lemma:boundEGDnorm} we have
\begin{align*}
    \Re \l \nabla f(\bU),\E[\nabla_{\bU}\ell_t(\bU)]\r&=~ \Re\l 4\left( \bU \bU^{\dag}- \brhos\right)\bU,\frac{B}{d}\left( \bU \bU^{\dag}- \brhos\right)\bU\r\cr
    &=~ \frac{4B}{d}\Re\l \left( \bU \bU^{\dag}- \brhos\right)\bU\bU^{\dag},\bU \bU^{\dag}- \brhos \r\cr
    &\ge~ \frac{2B}{d}  (1-\delta)^2\lambdas_r\lV\bU \bU^{\dag}- \brhos\rV_{\fro}^2.
\end{align*}
Now we prove \eqref{ineq:Egrad2}. Similar to \eqref{ineq:proveboundgd}, noticing that $\E[z_{t,k}^2|\bA_{t,k}]\le\frac{1}{L}\le\frac{\varepsilon_0^2}{d}$, we have
\begin{align}
    &\E\lV  \left[ \Tr(\bA_{t,k} \bU \bU^{\dag})-\Tr(\bA_{t,k} \brhos) -z_{t,k}\right]\bA_{t,k}\bU\rV_{\fro}^2\le~ \E \left[ \left(2\l\bA_{t,k}, \bU \bU^{\dag}- \brhos\r^2 +2 z_{t,k}^2\right)\lV\bA_{t,k}\bU\rV_{\fro}^2\right]\nonumber\\
    &\le ~\left[\frac{2}{d^2}\sum_{i=1}^{d^2} \lv\l\bW_i, \bU \bU^{\dag}-\brhos\r \rv^2+\frac{2\varepsilon_0^2}{d}\right]\lV\bU\bU^{\dag}\rV_{*}\nonumber\\
    &\le ~\left[\frac{2}{d}\sum_{i=1}^{d^2} \frac{1}{d}\lv\l\bW_i, \bU \bU^{\dag}-\brhos\r \rv^2+\frac{2\varepsilon_0^2}{d}\right]r\lV\bU\rV^2\nonumber\\
    &\le~ \left[\frac{2}{d}\lV  \bU \bU^{\dag}- \brhos\rV_{\fro}^2+\frac{2\varepsilon_0^2}{d}\right]r(\kappa+\delta)\lambdas_r,~ k\in [B]\label{ineq:boundEgd2}
\end{align}
where the last inequality follows from~\eqref{ineq:Uopup}. Thus, we have
 \begin{align*}
    &\E[\lV\nabla_{\bU}\ell_t(\bU)\rV_{\fro}^2] =~ \E\lV  \sum_{k=1}^B \left[ \Tr(\bm{A}_{t,k} \bU \bU^{\dag})-y_{t,k} \right]\bm{A}_{t,k}\bU\rV_{\fro}^2\cr
    &=~ \sum_{k=1}^B\E\lV  \left[ \Tr(\bm{A}_{t,k} \bU \bU^{\dag})-y_{t,k} \right]\bm{A}_{t,k}\bU\rV_{\fro}^2\cr
    &\quad +\sum_{k_1\neq k_2}\Bigl\langle \E\left[ (\Tr(\bm{A}_{t,k_1} \bU \bU^{\dag})-y_{t,k_1}) \bm{A}_{t,k_1}\bU\right], \E\left[ (\Tr(\bm{A}_{t,k_2} \bU \bU^{\dag})-y_{t,k_2}) \bm{A}_{t,k_2}\bU\right] \Bigr\rangle\cr
    &=~ \sum_{k=1}^B\E\lV  \left[ \Tr(\bm{A}_{t,k} \bU \bU^{\dag})-\Tr(\bA_{t,k} \brhos) -z_{t,k} \right]\bm{A}_{t,k}\bU\rV_{\fro}^2\cr
    &\quad +\sum_{k_1\neq k_2}\frac{1}{d^2}\l \left( \bU \bU^{\dag}- \brhos\right)\bU,\left( \bU \bU^{\dag}- \brhos\right)\bU\r\cr
    &\le~ B\left[\frac{2}{d}\lV  \bU \bU^{\dag}- \brhos\rV_{\fro}^2+\frac{2\varepsilon_0^2}{d}\right] r(\kappa+\delta)\lambdas_r +\frac{B^2-B}{d^2}\lV  \bU \bU^{\dag}- \brhos\rV_{\fro}^2(\kappa+\delta)\lambdas_r\cr
    &\le~ \left[\frac{2dB+B^2-B}{d^2}\lV  \bU \bU^{\dag}- \brhos\rV_{\fro}^2+\frac{2B\varepsilon_0^2}{d}\right] r(\kappa+\delta)\lambdas_r\cr
    &\le~ \left[\frac{3B\max\{d,B\}}{d^2}\lV  \bU \bU^{\dag}- \brhos\rV_{\fro}^2+\frac{2B\varepsilon_0^2}{d}\right] r(\kappa+\delta)\lambdas_r,
\end{align*}
    where the first inequality follows from \eqref{ineq:Uopup} and~\eqref{ineq:boundEgd2}. 
    
Finally, we prove \eqref{ineq:Egrad4}. Conditional Hoeffding gives $\P(|z_{t,k}|>v|\bA_{t,k})\le2e^{-\frac{Lv^2}{2}}$. Hence
\begin{align*}
    \E[z_{t,k}^4|\bA_{t,k}]&= 4\int_0^{\infty}v^3\P(|z_{t,k}|>v|\bA_{t,k})\,dv\cr
    &\le 8\int_0^{\infty}v^3e^{-\frac{Lv^2}{2}}\,dv=\frac{16}{L^2}\le\frac{\varepsilon_0^4}{d^2}.
\end{align*}
By Cauchy--Schwarz we have
\begin{align*}
    \E[\lV\nabla_{\bU}\ell_t(\bU)\rV_{\fro}^4]&=~ \E\lV \sum_{k=1}^B \left[ \Tr(\bA_{t,k} \bU \bU^{\dag})-\Tr(\bA_{t,k} \brhos) -z_{t,k}\right]\bA_{t,k}\bU\rV_{\fro}^4\cr
    &\le~ \E \biggl[  \left(\sum_{k=1}^B\left( \Tr(\bA_{t,k} \bU \bU^{\dag})-\Tr(\bA_{t,k} \brhos) -z_{t,k}\right)^2\right)^2 \left(\sum_{k=1}^B\lV\bA_{t,k}\bU\rV_{\fro}^2\right)^2\biggr]\cr
    &\le~ \E \biggl[ \Bigl(\sum_{k=1}^B\left(2\l\bA_{t,k}, \bU \bU^{\dag}- \brhos\r^2 +2z_{t,k}^2\right) B\lV\bU\bU^{\dag}\rV_{*}\Bigr)^2\biggr]\cr
    &\le~ 8B \E \left[ \sum_{k=1}^B\left(\l\bA_{t,k}, \bU \bU^{\dag}- \brhos\r^4 +\frac{\varepsilon_0^4}{d^2}\right) \right] B^2\lV\bU\bU^{\dag}\rV_{*}^2\cr
    &\le~ 8B^3 \E \Biggl[ \sum_{k=1}^B\Bigl(2r \lV\bU\bU^{\dag}-\brhos\rV_\fro^2 \l\bA_{t,k}, \bU \bU^{\dag}- \brhos\r^2 +\frac{\varepsilon_0^4}{d^2}\Bigr)\Biggr] \lV\bU\bU^{\dag}\rV_{*}^2\cr
    &\le~ 8B^4\Biggl[\frac{2r\lV\bU\bU^{\dag}-\brhos\rV_\fro^2}{d^2}\sum_{i=1}^{d^2} \lv\l\bW_i, \bU \bU^{\dag}-\brhos\r \rv^2 +\frac{\varepsilon_0^4}{d^2}\Biggr]r^2\lV\bU\rV^4\cr
    &\le~ 8 B^4\left(\frac{2r\lV\bU\bU^{\dag}-\brhos\rV_\fro^4}{d}+\frac{\varepsilon_0^4}{d^2 }\right)r^2(\kappa+\delta)^2(\lambdas_r)^2,
\end{align*}
where the fourth inequality follows from the fact that
\begin{align*}
    \l\bA_{t,k}, \bU \bU^{\dag}- \brhos\r^4&\le~ \lV\bA_{t,k} \rV^2\lV \bU \bU^{\dag}- \brhos\rV_{*}^2\l\bA_{t,k}, \bU \bU^{\dag}- \brhos\r^2
\end{align*}
and $\lV \bU \bU^{\dag}- \brhos\rV_{*}^2\le 2r\lV\bU\bU^{\dag}-\brhos\rV_\fro^2$. This completes the proof.
\end{proof}
\subsection{Proof of \texorpdfstring{\Cref{thm:Econtraction}}{Theorem \ref*{thm:Econtraction}}}\label{proofofEcontraction}
The estimations in \Cref{subsec:keylemmas} allow us to prove the local contraction property in expectation, and we present the proof of \Cref{thm:Econtraction} in the following.
\begin{proof}
We mainly use the local smoothness of $f(\bU)=\lV\bU\bU^{\dag}-\brhos\rV_\fro^2$ to derive the local contraction property. We prove the estimate for $e_t\le\frac{\lambdas_r}{2}$, which also covers the region in the theorem. By \eqref{ineq:Uopup}, we have $\lV\bU_t\rV^2\le\frac{3\lambdas_1}{2}$ and $\lV\bU_t\rV_\fro^2\le\frac{3r\lambdas_1}{2}$. For $B\le40\kappa^2$, the step-size bound gives
\begin{align*}
    \eta Br e_t&\le \frac{c_1B}{2\kappa^2}\le20c_1,\cr
    \frac{\eta^2Br}{L}&\le \frac{c_1^2B}{\kappa^2r(\lambdas_1)^2L}\le40c_1^2\frac{r}{L}\le40c_1^2.
\end{align*}
For $B\ge40\kappa^2$, we have
\begin{align*}
    \eta Br e_t&\le \frac{c_1}{2},\cr
    \frac{\eta^2Br}{L}&\le \frac{c_1^2\kappa^2}{Br(\lambdas_1)^2L}\le\frac{c_1^2r}{40L}\le\frac{c_1^2}{40},
\end{align*}
where we used $r\lambdas_1\ge1$ and $L\ge d\ge r$.

Let $\ell_t^0$ be obtained from $\ell_t$ by setting $y_{t,k}=\Tr(\bA_{t,k}\brhos)$. Conditioning on $\mathcal{F}_t$, \Cref{lemma:boundSG} gives
\begin{align*}
    \lV\nabla_{\bU}\ell_{t+1}^0(\bU_t)\rV_\fro&\le CBr\sqrt{\lambdas_1}e_t,\cr
    \E\left[\lV\nabla_{\bU}\ell_{t+1}^0(\bU_t)\rV_\fro^2\right]&\le \frac{CBr\lambdas_1}{d}e_t^2,\cr
    \E\left[\lV\nabla_{\bU}\ell_{t+1}^0(\bU_t)\rV_\fro^4\right]&\le CB^2r^2\lambdas_1e_t^2\E\left[\lV\nabla_{\bU}\ell_{t+1}^0(\bU_t)\rV_\fro^2\right]\cr
    &\le \frac{CB^3r^3(\lambdas_1)^2}{d}e_t^4,
\end{align*}
Define
\begin{align*}
    \bm{N}_{t+1}&:= \sum_{k=1}^B z_{t+1,k}\bA_{t+1,k}\bU_t,\cr
    \nabla_{\bU}\ell_{t+1}(\bU_t)&= \nabla_{\bU}\ell_{t+1}^0(\bU_t)-\bm{N}_{t+1}.
\end{align*}
Conditional on the past and the settings, the variables $z_{t+1,k}$ are independent and centered, with $\E[z_{t+1,k}^2]\le\frac{1}{L}$ and $\E[z_{t+1,k}^4]\le\frac{16}{L^2}$. Independence and Cauchy--Schwarz therefore give
\begin{align*}
    \E\lV\bm{N}_{t+1}\rV_\fro^2&=\sum_{k=1}^B\E[z_{t+1,k}^2]\lV\bU_t\rV_\fro^2\cr
    &\le \frac{B}{L}\lV\bU_t\rV_\fro^2\le\frac{3Br\lambdas_1}{2L},
\end{align*}
and
\begin{align*}
    \E\lV\bm{N}_{t+1}\rV_\fro^4&\le~ \Bigl(\sum_{k=1}^B\E[z_{t+1,k}^4]+3\sum_{k\ne h}\E[z_{t+1,k}^2]\E[z_{t+1,h}^2]\Bigr)\lV\bU_t\rV_\fro^4\cr
    &\le~ \frac{16B+3B(B-1)}{L^2}\lV\bU_t\rV_\fro^4\cr
    &\le~ 16\frac{B^2}{L^2}\lV\bU_t\rV_\fro^4\le36\left(\frac{Br\lambdas_1}{L}\right)^2.
\end{align*}
Consequently,
\begin{align*}
    \eta^2\E\left[\lV\nabla_{\bU}\ell_{t+1}(\bU_t)\rV_\fro^2\Big|\mathcal{F}_t\right]+\frac{\eta^4}{\lambdas_1}\E\left[\lV\nabla_{\bU}\ell_{t+1}(\bU_t)\rV_\fro^4\Big|\mathcal{F}_t\right]&\le~ \frac{C\eta^2Br\lambdas_1}{d}\left(e_t^2+\varepsilon_0^2\right),
\end{align*}
where $C>0$ is a numerical constant. Here, we used $\eta Br e_t\le C$, $\frac{\eta^2Br}{L}\le C$, and $\frac{1}{L}\le\frac{\varepsilon_0^2}{112d\log d}$.

By the expansion in \Cref{lemma:Lip} and Young's inequality, we have, for every displacement $\bV$,
\begin{align*}
    \lv f(\bU_t+\bV)-f(\bU_t)-\Re\l\nabla f(\bU_t),\bV\r\rv&\le~ \left(2e_t+4\lV\bU_t\rV\lV\bV\rV_\fro+4\lV\bU_t\rV^2+\lV\bV\rV_\fro^2\right)\lV\bV\rV_\fro^2\cr
    &\le~ C\left(\lambdas_1\lV\bV\rV_\fro^2+\lV\bV\rV_\fro^4\right).
\end{align*}
Thus, by the update rule and \eqref{eq:Egrad}, we have
\begin{align*}
    \E[e_{t+1}^2|\mathcal{F}_t]&= \E\left[f\left(\bU_t-\eta\nabla_{\bU}\ell_{t+1}(\bU_t)\right)\Big|\mathcal{F}_t\right]\cr
    &\le e_t^2-\frac{4\eta B}{d}\lV\left(\bU_t\bU_t^{\dag}-\brhos\right)\bU_t\rV_\fro^2 +\frac{C\eta^2Br(\lambdas_1)^2}{d}\left(e_t^2+\varepsilon_0^2\right).
\end{align*}
The stronger estimate at the end of the proof of \Cref{lemma:boundEGDnorm}, with $\delta=\frac{1}{2}$, gives $\lV(\bU_t\bU_t^{\dag}-\brhos)\bU_t\rV_\fro^2\ge\frac{3\lambdas_r}{16}e_t^2$. Choosing $c_1$ sufficiently small, the preceding inequality implies
\begin{align*}
    \E[e_{t+1}^2|\mathcal{F}_t]&\le \left(1-\frac{\eta B\lambdas_r}{2d}\right)e_t^2+\frac{\eta B\lambdas_r\varepsilon_0^2}{8d}.
\end{align*}
This completes the proof.
\end{proof}
{
\subsection{Proof of \texorpdfstring{\Cref{col:appLowrank}}{Corollary \ref*{col:appLowrank}}}\label{proofofappLowrank}
\begin{proof}
For each $\bA_t$, similar to \eqref{eq:Egrad}, we have
\begin{align*}
    \E[\nabla_{\bU}\ell_t(\bU)]&= \frac{1}{d}\left(\bU\bU^{\dag}-\brhos\right)\bU.
\end{align*}
Let $\bU\bU^{\dag}=\bV\bSigma\bV^{\dag}$ be its compact singular value decomposition and let $\bm{V}_r$ contain the leading $r$ eigenvectors of $\brhos$. The entrance condition implies $\lambda_r(\bSigma)\ge\frac{2\lambdas_r}{3}$ and $\lV(\bI-\bV\bV^{\dag})\bm{V}_r\rV\le\frac{1}{3}$. Applying the block estimate in the proof of \Cref{lemma:boundEGDnorm}, we obtain
\begin{align*}
    \lV\left(\bU\bU^{\dag}-\brho_r\right)\bU\rV_\fro^2&\ge \frac{8\lambdas_r}{27}\lV\bU\bU^{\dag}-\brho_r\rV_\fro^2.
\end{align*}
For $f_r(\bU)=\lV\bU\bU^{\dag}-\brho_r\rV_\fro^2$, Young's inequality gives
\begin{align*}
    \Re\l\nabla f_r(\bU),\E[\nabla_{\bU}\ell_t(\bU)]\r&=~ \frac{4}{d}\lV(\bU\bU^{\dag}-\brho_r)\bU\rV_\fro^2-\frac{4}{d}\Re\l(\bU\bU^{\dag}-\brho_r)\bU,\bm{E}\bU\r\cr
    &\ge~ \frac{3}{d}\lV(\bU\bU^{\dag}-\brho_r)\bU\rV_\fro^2-\frac{4}{d}\lV\bm{E}\bU\rV_\fro^2\cr
    &\ge~ \frac{8\lambdas_r}{9d}\lV\bU\bU^{\dag}-\brho_r\rV_\fro^2-\frac{16\lambdas_1}{3d}\lV\bm{E}\rV_\fro^2,
\end{align*}
where we used $\lV\bU\rV^2\le\frac{4\lambdas_1}{3}$. As in \eqref{ineq:boundEgd2}, we have
\begin{align*}
    \E[\lV\nabla_{\bU}\ell_t(\bU)\rV_\fro^2]&\le~ \left[\frac{3}{d}\lV\bU\bU^{\dag}-\brho_r\rV_\fro^2+\frac{3}{d}\lV\bm{E}\rV_\fro^2+\frac{3\varepsilon_0^2}{d}\right]\frac{4r\lambdas_1}{3}.
\end{align*}
The entrance condition also gives $\lV\bm{E}\rV_\fro\le\frac{\lambdas_r}{3}$ by optimality of rank-$r$ truncation and $\lV\bU\bU^{\dag}-\brho_r\rV_\fro\le\frac{2\lambdas_r}{3}$. Since $r\lambdas_r\le1$, $\lV\bm{E}\rV_*\le1$, and $|z_t|\le2$, we have
\begin{align*}
    \eta^2\lV\nabla_{\bU}\ell_t(\bU)\rV_\fro^2&=~ \eta^2\left|\Tr(\bA_t(\bU\bU^{\dag}-\brho_r))-\Tr(\bA_t\bm{E})-z_t\right|^2\lV\bU\rV_\fro^2\cr
    &\le~ \eta^2\left(\frac{4r\lambdas_r}{3}+1+2\right)^2\frac{4r\lambdas_1}{3}\le26\eta^2r\lambdas_1\le\lambdas_1,
\end{align*}
provided $c_1$ is sufficiently small. Using the expansion in the proof of \Cref{lemma:Lip}, we obtain
\begin{align*}
    \Bigl|f_r(\bU-\eta\nabla_{\bU}\ell_t(\bU))-f_r(\bU) +\eta\Re\l\nabla f_r(\bU),\nabla_{\bU}\ell_t(\bU)\r\Bigr|&\le~ \left(\frac{4}{3}+\frac{8}{\sqrt3}+\frac{16}{3}+1\right)\eta^2\lambdas_1\lV\nabla_{\bU}\ell_t(\bU)\rV_\fro^2\cr
    &\le~ 13\eta^2\lambdas_1\lV\nabla_{\bU}\ell_t(\bU)\rV_\fro^2.
\end{align*}
Combining the preceding bounds gives
\begin{align*}
    \E\left[\lV\bU_{t+1}\bU_{t+1}^{\dag}-\brho_r\rV_\fro^2\Big|\mathcal{F}_t\right]&\le~ \left(1-\frac{\eta\lambdas_r}{2d}\right)\lV\bU_t\bU_t^{\dag}-\brho_r\rV_\fro^2+\frac{16\eta\lambdas_1}{3d}\lV\bm{E}\rV_\fro^2 +\frac{\eta\lambdas_r}{6d}\left(\lV\bm{E}\rV_\fro^2+\varepsilon_0^2\right),
\end{align*}
provided $c_1$ is sufficiently small. Finally, $\lambdas_1\le1$ and $\lV\bm{E}\rV_\fro\le\frac{\lambdas_r}{3}$ imply $\frac{16\lambdas_1}{3}\lV\bm{E}\rV_\fro^2\le\frac{16\lambdas_r}{9}\lV\bm{E}\rV_\fro$. Substituting this inequality proves the corollary.
\end{proof}
}
\medskip

\subsection{Proof of \texorpdfstring{\Cref{thm:propconvergence}}{Theorem \ref*{thm:propconvergence}}}\label{subsec:proofofprobconvergence}
The proof follows from the Azuma--Bernstein inequality. Let $\ell_t^0$ be the noiseless loss used in the proof of \Cref{thm:Econtraction}, and define
\begin{align*}
    \bm{N}_{t+1}&= \sum_{k=1}^Bz_{t+1,k}\bA_{t+1,k}\bU_t,\cr
    \nabla_{\bU}\ell_{t+1}(\bU_t)&= \nabla_{\bU}\ell_{t+1}^0(\bU_t)-\bm{N}_{t+1},\cr
    \mathcal{H}_{t+1}&= \mathcal{F}_t\vee\sigma\left(\{\bA_{t+1,k}\}_{k=1}^B\right).
\end{align*}
Conditional on $\mathcal{H}_{t+1}$, independence and centering of the shots give
\begin{align*}
    \E[\bm{N}_{t+1}|\mathcal{H}_{t+1}]&= 0,\cr
    \E[\lV\bm{N}_{t+1}\rV_\fro^2|\mathcal{H}_{t+1}]&\le \frac{B}{L}\lV\bU_t\rV_\fro^2.
\end{align*}
Changing one shot changes $\lV\bm{N}_{t+1}\rV_\fro$ by at most $\frac{2}{L}\lV\bU_t\rV_\fro$. The bounded differences inequality gives
\begin{align*}
    \P\Biggl(\lV\bm{N}_{t+1}\rV_\fro> \sqrt{\frac{B}{L}}\left(1+\sqrt{20\log d}\right)\lV\bU_t\rV_\fro\Big|\mathcal{H}_{t+1}\Biggr)&\le~ \exp\left(-\frac{2(20B\log d)\lV\bU_t\rV_\fro^2/L}{BL(2\lV\bU_t\rV_\fro/L)^2}\right)=d^{-10}.
\end{align*}
For $e_t\le\frac{\lambdas_r}{2}$, the shot budget implies
\begin{align*}
    \frac{B}{L}\left(1+\sqrt{20\log d}\right)^2\lV\bU_t\rV_\fro^2&\le~ \frac{3Br\lambdas_1\varepsilon_0^2}{224d\log d}\left(1+\sqrt{20\log d}\right)^2\le\frac{Br\lambdas_1\varepsilon_0^2}{d}.
\end{align*}
Consequently,
\begin{align}
    \P\left(\lV\bm{N}_{t+1}\rV_\fro>\sqrt{\frac{Br\lambdas_1}{d}}\varepsilon_0\Big|\mathcal{F}_t\right)&\le d^{-10}.\label{ineq:shotnorm-local}
\end{align}
Define $\mathfrak{E}_t$ to be the event that this noise bound holds through round $t$ and, for all $0\le\tau\le t$,
\begin{align*}
    e_\tau^2&\le 2\left(1-\frac{\eta B\lambdas_r}{4d}\right)^\tau e_0^2 +\left[1-\left(1-\frac{\eta B\lambdas_r}{4d}\right)^\tau\right]\varepsilon_0^2.
\end{align*}
Since $2e_0^2\le\frac{2(\lambdas_r)^2}{9}$ and $\varepsilon_0^2\le\frac{(\lambdas_r)^2}{4}$, this event implies $e_t\le\frac{\lambdas_r}{2}$. We condition on the initial factor throughout the proof.

With expectations conditional on $\mathcal{F}_t$, define
\begin{align*}
    \Delta_{t+1}&:= \eta\Re\Bigl\langle\nabla f(\bU_t), \E[\nabla_{\bU}\ell_{t+1}(\bU_t)]-\nabla_{\bU}\ell_{t+1}(\bU_t)\Bigr\rangle\cr
    &\quad +16\eta^2\lambdas_1\Bigl(\lV\nabla_{\bU}\ell_{t+1}^0(\bU_t)\rV_\fro^2 -\E[\lV\nabla_{\bU}\ell_{t+1}^0(\bU_t)\rV_\fro^2]\Bigr).
\end{align*}
Then $\E[\Delta_{t+1}|\mathcal{F}_t]=0$.

\begin{lemma}\label{lemma:supermart}
Under the assumptions of \Cref{thm:propconvergence}, define
\begin{align*}
    F_0&= e_0^2-\frac{\varepsilon_0^2}{4},\cr
    F_t&= F_0+\sum_{j=1}^t\left(1-\frac{\eta B\lambdas_r}{2d}\right)^{-j}\Delta_j1_{\mathfrak{E}_{j-1}}.
\end{align*}
Then $F_t$ is a supermartingale. On $\mathfrak{E}_{t-1}$ and the noise event in round $t$, we have
\begin{align*}
    e_t^2&\le \left(1-\frac{\eta B\lambdas_r}{2d}\right)^tF_t+\frac{\varepsilon_0^2}{4}.
\end{align*}
\end{lemma}
\begin{proof}
By conditional centering, we have
\begin{align*}
    \E[F_t|\mathcal{F}_{t-1}]&= F_{t-1}.
\end{align*}
We next prove the error comparison. Pauli orthogonality and \eqref{ineq:Uopup} give
\begin{align}
    \E[\lV\nabla_{\bU}\ell_{t+1}^0(\bU_t)\rV_\fro^2|\mathcal{F}_t]&=~ \frac{B}{d}e_t^2\lV\bU_t\rV_\fro^2 +\frac{B(B-1)}{d^2}\lV(\bU_t\bU_t^{\dag}-\brhos)\bU_t\rV_\fro^2\nonumber\\
    &\le~ \frac{3Br\lambdas_1}{d}e_t^2.\label{ineq:Egrad22}
\end{align}
On the noise event, the displacement $\bV=-\eta\nabla_{\bU}\ell_{t+1}^0(\bU_t)+\eta\bm{N}_{t+1}$ satisfies
\begin{align*}
    \lV\bV\rV_\fro&\le \sqrt3\eta Br\sqrt{\lambdas_1}e_t+\eta\sqrt{\frac{Br\lambdas_1}{d}}\varepsilon_0\cr
    &\le \left(20\sqrt3+4\right)\eta\kappa r\lambdas_1\sqrt{\lambdas_1}\le\frac{\sqrt{\lambdas_1}}{8},
\end{align*}
provided $c_2$ is sufficiently small. The expansion in \Cref{lemma:Lip} yields
\begin{align*}
    \lv f(\bU_t+\bV)-f(\bU_t)-\Re\l\nabla f(\bU_t),\bV\r\rv&\le~ \Bigl(2e_t+4\lV\bU_t\rV\lV\bV\rV_\fro+4\lV\bU_t\rV^2+\lV\bV\rV_\fro^2\Bigr)\lV\bV\rV_\fro^2\cr
    &\le~ 8\lambdas_1\lV\bV\rV_\fro^2\cr
    &\le~ 16\eta^2\lambdas_1\left(\lV\nabla_{\bU}\ell_{t+1}^0(\bU_t)\rV_\fro^2+\lV\bm{N}_{t+1}\rV_\fro^2\right).
\end{align*}
By \eqref{ineq:Egrad22} and the stronger bound at the end of the proof of \Cref{lemma:boundEGDnorm}, we have
\begin{align*}
    e_{t+1}^2&\le e_t^2-\frac{4\eta B}{d}\lV(\bU_t\bU_t^{\dag}-\brhos)\bU_t\rV_\fro^2+\Delta_{t+1} +16\eta^2\lambdas_1\E[\lV\nabla_{\bU}\ell_{t+1}^0(\bU_t)\rV_\fro^2|\mathcal{F}_t] +\frac{16\eta^2Br(\lambdas_1)^2}{d}\varepsilon_0^2\cr
    &\le \left(1-\frac{3\eta B\lambdas_r}{4d}+\frac{48\eta^2Br(\lambdas_1)^2}{d}\right)e_t^2 +\frac{16\eta^2Br(\lambdas_1)^2}{d}\varepsilon_0^2+\Delta_{t+1}\cr
    &\le \left(1-\frac{\eta B\lambdas_r}{2d}\right)e_t^2 +\frac{\eta B\lambdas_r\varepsilon_0^2}{8d}+\Delta_{t+1},
\end{align*}
where the last inequality holds for $\eta\le\frac{1}{384\kappa r\lambdas_1}$. Subtracting $\frac{\varepsilon_0^2}{4}$ and iterating gives
\begin{align*}
    \left(1-\frac{\eta B\lambdas_r}{2d}\right)^{-t}\left(e_t^2-\frac{\varepsilon_0^2}{4}\right)&\le~ e_0^2-\frac{\varepsilon_0^2}{4} +\sum_{j=1}^t\left(1-\frac{\eta B\lambdas_r}{2d}\right)^{-j}\Delta_j=F_t.
\end{align*}
This completes the proof.
\end{proof}

\begin{lemma}\label{lemma:boundvar}
Under the assumptions of \Cref{thm:propconvergence}, there exist numerical constants $c_3,c_4>0$ such that, for $e_t\le\frac{\lambdas_r}{2}$,
\begin{align}
    \E[e^{\theta\Delta_{t+1}}|\mathcal{F}_t]&\le \exp\left(\frac{\theta^2v_{t+1}}{2(1-\theta R_{t+1}/3)}\right),\label{ineq:EF-F}
\end{align}
whenever $\theta\ge0$ and $\theta R_{t+1}<3$, where
\begin{align*}
    R_{t+1}&= c_3\eta Br\lambdas_1e_t^2,\cr
    v_{t+1}&= c_4\eta^2Br(\lambdas_1)^2\left(\frac{e_t^4}{d}+\frac{e_t^2}{L}\right).
\end{align*}
\end{lemma}
\begin{proof}
All expectations in this proof condition on $\mathcal{F}_t$, unless conditioning on $\mathcal{H}_{t+1}$ is specified. Let $\overline\Delta_{t+1}=\E[\Delta_{t+1}|\mathcal{H}_{t+1}]$. For the linear term, Cauchy--Schwarz and \eqref{ineq:boundgd} with zero shot noise give
\begin{align*}
    \lv\Re\l\nabla f(\bU_t),\nabla_{\bU}\ell_{t+1}^0(\bU_t)\r\rv&\le~ 4\lV(\bU_t\bU_t^{\dag}-\brhos)\bU_t\rV_\fro\lV\nabla_{\bU}\ell_{t+1}^0(\bU_t)\rV_\fro\cr
    &\le~ 12Br\lambdas_1e_t^2.
\end{align*}
For the quadratic term, Cauchy--Schwarz gives
\begin{align*}
    \lV\nabla_{\bU}\ell_{t+1}^0(\bU_t)\rV_\fro^2&\le~ B\lV\bU_t\rV_\fro^2\sum_{k=1}^B\lv\l\bA_{t+1,k},\bU_t\bU_t^{\dag}-\brhos\r\rv^2\cr
    &\le~ 2B^2re_t^2\lV\bU_t\rV_\fro^2\le3B^2r^2\lambdas_1e_t^2.
\end{align*}
By \eqref{ineq:Egrad22},
\begin{align*}
    \E[\lV\nabla_{\bU}\ell_{t+1}^0(\bU_t)\rV_\fro^4]&\le3B^2r^2\lambdas_1e_t^2\E[\lV\nabla_{\bU}\ell_{t+1}^0(\bU_t)\rV_\fro^2]\cr
    &\le \frac{9B^3r^3(\lambdas_1)^2}{d}e_t^4.
\end{align*}
Consequently, since $\eta\le\frac{C}{Br\lambdas_1}$, we have
\begin{align*}
    |\overline\Delta_{t+1}|&\le C\eta Br\lambdas_1e_t^2 +C\eta^2B^2r^2(\lambdas_1)^2e_t^2\le R_{t+1},\cr
    \E[\overline\Delta_{t+1}]&= 0.
\end{align*}
By Cauchy--Schwarz and the preceding moment estimates, we also have
\begin{align}
    \E[\overline\Delta_{t+1}^2]&\le 2\eta^2\lV\nabla f(\bU_t)\rV_\fro^2 \E[\lV\nabla_{\bU}\ell_{t+1}^0(\bU_t)\rV_\fro^2] +512\eta^4(\lambdas_1)^2\E[\lV\nabla_{\bU}\ell_{t+1}^0(\bU_t)\rV_\fro^4]\nonumber\\
    &\le \frac{C\eta^2Br(\lambdas_1)^2}{d} \left[1+\left(\eta Br\lambdas_1\right)^2\right]e_t^4\nonumber\\
    &\le \frac{C\eta^2Br(\lambdas_1)^2}{d}e_t^4.\label{ineq:EgradEG}
\end{align}
For the shot-noise term, we have
\begin{align*}
    \Delta_{t+1}-\overline\Delta_{t+1}&= \eta\Re\l\nabla f(\bU_t),\bm{N}_{t+1}\r\cr
    &= \eta\sum_{k=1}^B z_{t+1,k}\Re\l\nabla f(\bU_t),\bA_{t+1,k}\bU_t\r.
\end{align*}
By Cauchy--Schwarz,
\begin{align*}
    \lv\Re\l\nabla f(\bU_t),\bA_{t+1,k}\bU_t\r\rv&\le~ 4\lV(\bU_t\bU_t^{\dag}-\brhos)\bU_t\rV_\fro\lV\bU_t\rV_\fro\le6\sqrt r\lambdas_1e_t.
\end{align*}
Conditional Hoeffding therefore gives
\begin{align*}
    \E[e^{\theta(\Delta_{t+1}-\overline\Delta_{t+1})}|\mathcal{H}_{t+1}]&=~ \prod_{k=1}^B\E\left[\exp\left(\theta\eta z_{t+1,k}\Re\l\nabla f(\bU_t),\bA_{t+1,k}\bU_t\r\right)\Big|\mathcal{H}_{t+1}\right]\cr
    &\le~ \exp\left(\frac{\theta^2\eta^2}{2L}\sum_{k=1}^B\left(\Re\l\nabla f(\bU_t),\bA_{t+1,k}\bU_t\r\right)^2\right)\cr
    &\le~ \exp\left(\frac{18\theta^2\eta^2Br(\lambdas_1)^2e_t^2}{L}\right).
\end{align*}
For the centered settings term, using $k!\ge2\cdot3^{k-2}$ and the exponential series, we have
\begin{align*}
    \E[e^{\theta\overline\Delta_{t+1}}]&\le 1+\sum_{k=2}^{\infty}\frac{\theta^kR_{t+1}^{k-2}}{k!} \E[\overline\Delta_{t+1}^2]\cr
    &\le \exp\left(\frac{\theta^2\E[\overline\Delta_{t+1}^2]}{2(1-\theta R_{t+1}/3)}\right).
\end{align*}
Taking conditional expectations over the settings and using \eqref{ineq:EgradEG}, we obtain
\begin{align*}
    \E[e^{\theta\Delta_{t+1}}]&\le \exp\left(\frac{18\theta^2\eta^2Br(\lambdas_1)^2e_t^2}{L}\right)\E[e^{\theta\overline\Delta_{t+1}}]\cr
    &\le \exp\left(\frac{\theta^2v_{t+1}}{2(1-\theta R_{t+1}/3)}\right),
\end{align*}
after increasing $c_4$. This completes the proof.
\end{proof}

\begin{lemma}\label{lemma:propconverge}
Under the assumptions of \Cref{thm:propconvergence}, we have
\begin{align}
    \P(\mathfrak{E}_{t-1}\cap\mathfrak{E}_t^c)&\le \frac{3}{d^{10}},\quad t\in[T],\label{event:propconverge}
\end{align}
provided $c_2$ is sufficiently small.
\end{lemma}
\begin{proof}
Let $R$ and $\lambda^2$ be deterministic upper bounds satisfying
\begin{align*}
    \max_{1\le j\le t}\left(1-\frac{\eta B\lambdas_r}{2d}\right)^{-j}R_j1_{\mathfrak{E}_{j-1}}&\le R,\cr
    \sum_{j=1}^t\left(1-\frac{\eta B\lambdas_r}{2d}\right)^{-2j}v_j1_{\mathfrak{E}_{j-1}}&\le \lambda^2.
\end{align*}
By \eqref{ineq:EF-F}, successive conditioning gives
\begin{align*}
    \E\exp\Biggl[\theta(F_t-F_0) -\frac{\theta^2}{2(1-\theta R/3)} \sum_{j=1}^t\left(1-\frac{\eta B\lambdas_r}{2d}\right)^{-2j}v_j1_{\mathfrak{E}_{j-1}}\Biggr]&\le1,\cr
    \E[e^{\theta(F_t-F_0)}]&\le\exp\left(\frac{\theta^2\lambda^2}{2(1-\theta R/3)}\right).
\end{align*}
Markov's inequality, with $\theta=\frac{\beta}{\lambda^2+R\beta/3}$, yields the Azuma--Bernstein bound
\begin{align*}
    \P(F_t>F_0+\beta)&\le \exp\left(-\theta\beta+\frac{\theta^2\lambda^2}{2(1-\theta R/3)}\right)\cr
    &\le \exp\left(-\frac{\beta^2}{2\lambda^2+\frac{2}{3}R\beta}\right).
\end{align*}
We next estimate these quantities by geometric summation. For $j\le t$, we have
\begin{align*}
    \sum_{j=1}^t\left(\frac{1-\frac{\eta B\lambdas_r}{2d}}{1-\frac{\eta B\lambdas_r}{4d}}\right)^{2(t-j)}&\le \frac{4d}{\eta B\lambdas_r},\cr
    1-\left(1-\frac{\eta B\lambdas_r}{4d}\right)^t&\ge \frac{1}{8}\min\left\{\frac{\eta B\lambdas_r t}{d},1\right\}.
\end{align*}
Inserting the bound defining $\mathfrak{E}_{j-1}$ and summing the geometric series, we obtain
\begin{align*}
    \sum_{j=1}^t\left(1-\frac{\eta B\lambdas_r}{2d}\right)^{2t-2j}e_{j-1}^4 1_{\mathfrak{E}_{j-1}}&\le~ \frac{Cd}{\eta B\lambdas_r}\Bigg[\left(1-\frac{\eta B\lambdas_r}{4d}\right)^{2t}e_0^4 +\left[1-\left(1-\frac{\eta B\lambdas_r}{4d}\right)^t\right]^2\varepsilon_0^4\Bigg],
\end{align*}
\begin{align*}
    &\sum_{j=1}^t\left(1-\frac{\eta B\lambdas_r}{2d}\right)^{2t-2j}e_{j-1}^2 1_{\mathfrak{E}_{j-1}}\cr
    &\le~ \frac{Cd}{\eta B\lambdas_r}\Bigg[\left(1-\frac{\eta B\lambdas_r}{4d}\right)^t e_0^2 \min\left\{\frac{\eta B\lambdas_r t}{d},1\right\} +\varepsilon_0^2\min\left\{\left(\frac{\eta B\lambdas_r t}{d}\right)^2,1\right\}\Bigg]\cr
    &\le~ \frac{Cd}{\eta B\lambdas_r\varepsilon_0^2}\Bigg[\left(1-\frac{\eta B\lambdas_r}{4d}\right)^{2t}e_0^4 +\left[1-\left(1-\frac{\eta B\lambdas_r}{4d}\right)^t\right]^2\varepsilon_0^4\Bigg].
\end{align*}
Here, the last inequality uses $2ab\le a^2+b^2$ and the preceding lower bound for the geometric difference. Since $\frac{d}{L\varepsilon_0^2}\le\frac{1}{112\log d}$, we may choose
\begin{align*}
    \lambda^2&= \widetilde c_4\eta\kappa r\lambdas_1\left(1-\frac{\eta B\lambdas_r}{2d}\right)^{-2t} \Bigg[2\left(1-\frac{\eta B\lambdas_r}{4d}\right)^t e_0^2 +\left[1-\left(1-\frac{\eta B\lambdas_r}{4d}\right)^t\right]\varepsilon_0^2\Bigg]^2.
\end{align*}
For the increment bound, the definition of $\mathfrak{E}_{j-1}$ gives
\begin{align*}
    \max_{1\le j\le t}\left(1-\frac{\eta B\lambdas_r}{2d}\right)^{t-j}e_{j-1}^2 1_{\mathfrak{E}_{j-1}}&\le~ 2\left(1-\frac{\eta B\lambdas_r}{4d}\right)^{t-1}e_0^2 +\left[1-\left(1-\frac{\eta B\lambdas_r}{4d}\right)^t\right]\varepsilon_0^2.
\end{align*}
Using $R_j=c_3\eta Br\lambdas_1e_{j-1}^2$, $B\le40\kappa$, and $1-\frac{\eta B\lambdas_r}{4d}\ge\frac{1}{2}$, we may choose
\begin{align*}
    R&=\widetilde c_3\eta\kappa r\lambdas_1\left(1-\frac{\eta B\lambdas_r}{2d}\right)^{-t} \Bigg[2\left(1-\frac{\eta B\lambdas_r}{4d}\right)^t e_0^2 +\left[1-\left(1-\frac{\eta B\lambdas_r}{4d}\right)^t\right]\varepsilon_0^2\Bigg],
\end{align*}
where $\widetilde c_3=80c_3$.
For $\beta=\sqrt{20\lambda^2\log d}+\frac{20}{3}R\log d$, the probability is at most $d^{-10}$. Since $\eta\le\frac{c_2}{\kappa r\lambdas_1\log d}$, the preceding estimates imply
\begin{align*}
    \left(1-\frac{\eta B\lambdas_r}{2d}\right)^t\beta&\le\left(\sqrt{20\widetilde c_4c_2}+\frac{20\widetilde c_3c_2}{3}\right) \Bigg[2\left(1-\frac{\eta B\lambdas_r}{4d}\right)^t e_0^2 +\left[1-\left(1-\frac{\eta B\lambdas_r}{4d}\right)^t\right]\varepsilon_0^2\Bigg]\cr
    &\le~ \left(1-\frac{\eta B\lambdas_r}{4d}\right)^t e_0^2 +\frac{1}{2}\left[1-\left(1-\frac{\eta B\lambdas_r}{4d}\right)^t\right]\varepsilon_0^2,
\end{align*}
provided $c_2$ is sufficiently small. Moreover,
\begin{align*}
    \left(1-\frac{\eta B\lambdas_r}{2d}\right)^t e_0^2+\frac{1}{4}\left[1-\left(1-\frac{\eta B\lambdas_r}{2d}\right)^t\right]\varepsilon_0^2&\le~ \left(1-\frac{\eta B\lambdas_r}{4d}\right)^t e_0^2+\frac{1}{2}\left[1-\left(1-\frac{\eta B\lambdas_r}{4d}\right)^t\right]\varepsilon_0^2.
\end{align*}
Combining this inequality with \Cref{lemma:supermart} gives
\begin{align*}
    e_t^2&\le \left(1-\frac{\eta B\lambdas_r}{2d}\right)^t e_0^2 +\frac{1}{4}\left[1-\left(1-\frac{\eta B\lambdas_r}{2d}\right)^t\right]\varepsilon_0^2 +\left(1-\frac{\eta B\lambdas_r}{2d}\right)^t\beta\cr
    &\le 2\left(1-\frac{\eta B\lambdas_r}{4d}\right)^t e_0^2 +\left[1-\left(1-\frac{\eta B\lambdas_r}{4d}\right)^t\right]\varepsilon_0^2
\end{align*}
on $\mathfrak{E}_{t-1}$ and the round-$t$ noise event, except with probability $d^{-10}$. Combining this bound with \eqref{ineq:shotnorm-local} proves \eqref{event:propconverge}.
\end{proof}

\begin{proof}[Proof of \Cref{thm:propconvergence}]
By \Cref{lemma:propconverge}, we have
\begin{align*}
    \P(\mathfrak{E}_T^c)&\le \sum_{t=1}^T\P(\mathfrak{E}_{t-1}\cap\mathfrak{E}_t^c)\le\frac{3T}{d^{10}}.
\end{align*}
On $\mathfrak{E}_T$, the desired error bound holds for all $t\in[T]$. This completes the proof.
\end{proof}

\vspace{1em}

\noindent{\bf Acknowledgements.} This work is partially supported by the National Key Research and Development Program of China (No. 2024YFE0102500); the Fundamental Research Funds for the Central Universities under Grants 2042025kf0023, 2042026kf0061, 2042026fg0009; the National Natural Science Foundation of China (Nos. 12401121, 12371424, U24A2002, 12561160122, 62302346, 12622127), and the Natural Science Foundation of Hubei Province (2024AFE006, 2024AFA045), the Science and Technology Innovation Platform Program of Hubei Province (No. 2025CSA119), the Shenzhen Municipal Stability Support Program Project (No. 8960117/0123), and the HKRGC GRFs 16310926, 16307325, 16306124, and 16307023.
\bibliographystyle{IEEEtran}
{\small
\bibliography{OnlineQST}

\begin{thebibliography}{10}
\providecommand{\url}[1]{#1}
\csname url@samestyle\endcsname
\providecommand{\newblock}{\relax}
\providecommand{\bibinfo}[2]{#2}
\providecommand{\BIBentrySTDinterwordspacing}{\spaceskip=0pt\relax}
\providecommand{\BIBentryALTinterwordstretchfactor}{4}
\providecommand{\BIBentryALTinterwordspacing}{\spaceskip=\fontdimen2\font plus
\BIBentryALTinterwordstretchfactor\fontdimen3\font minus
  \fontdimen4\font\relax}
\providecommand{\BIBforeignlanguage}[2]{{%
\expandafter\ifx\csname l@#1\endcsname\relax
\typeout{** WARNING: IEEEtran.bst: No hyphenation pattern has been}%
\typeout{** loaded for the language `#1'. Using the pattern for}%
\typeout{** the default language instead.}%
\else
\language=\csname l@#1\endcsname
\fi
#2}}
\providecommand{\BIBdecl}{\relax}
\BIBdecl

\bibitem{PhysRevLett.113.040503}
C.~Schwemmer, G.~T\'oth, A.~Niggebaum, T.~Moroder, D.~Gross, O.~G\"uhne, and
  H.~Weinfurter, ``Experimental comparison of efficient tomography schemes for
  a six-qubit state,'' \emph{Physical Review Letters}, vol. 113, no.~4, 2014,
  art. no.~040503.

\bibitem{Riofro2017}
C.~A. Riofr{\'i}o, D.~Gross, S.~T. Flammia, T.~Monz, D.~Nigg, R.~Blatt, and
  J.~Eisert, ``Experimental quantum compressed sensing for a seven-qubit
  system,'' \emph{Nature Communications}, vol.~8, 2017, art. no.~15305.

\bibitem{7956181}
J.~Haah, A.~W. Harrow, Z.~Ji, X.~Wu, and N.~Yu, ``Sample-optimal tomography of
  quantum states,'' \emph{IEEE Transactions on Information Theory}, vol.~63,
  no.~9, pp. 5628--5641, 2017.

\bibitem{10.1145/2897518.2897544}
R.~O'Donnell and J.~Wright, ``Efficient quantum tomography,'' in
  \emph{Proceedings of the 48th Annual ACM Symposium on Theory of Computing},
  2016, pp. 899--912.

\bibitem{wrightthesis}
J.~Wright, ``How to learn a quantum state,'' Ph.D. dissertation, Carnegie
  Mellon University, Pittsburgh, PA, USA, 2016.

\bibitem{Yuen2023improvedsample}
H.~Yuen, ``An improved sample complexity lower bound for (fidelity) quantum
  state tomography,'' \emph{Quantum}, vol.~7, 2023, art. no.~890.

\bibitem{doi:10.1137/1.9781611977554.ch47}
J.~van Apeldoorn, A.~Cornelissen, A.~Gily{\'e}n, and G.~Nannicini, ``Quantum
  tomography using state-preparation unitaries,'' in \emph{Proceedings of the
  2023 Annual ACM-SIAM Symposium on Discrete Algorithms (SODA)}, 2023, pp.
  1265--1318.

\bibitem{10.1145/3618260.3649704}
S.~Chen, J.~Li, and A.~Liu, ``An optimal tradeoff between entanglement and copy
  complexity for state tomography,'' in \emph{Proceedings of the 56th Annual
  ACM Symposium on Theory of Computing}, 2024, pp. 1331--1342.

\bibitem{hu2024sampleoptimalmemoryefficient}
Y.~Hu, E.~Cervero-Mart{\'i}n, E.~Theil, L.~Man{\v{c}}inska, and M.~Tomamichel,
  ``Sample-optimal and memory-efficient quantum state tomography,''
  \emph{Physical Review A}, vol. 113, no.~5, 2026, art. no.~052446.

\bibitem{KUENG201788}
R.~Kueng, H.~Rauhut, and U.~Terstiege, ``Low rank matrix recovery from rank one
  measurements,'' \emph{Applied and Computational Harmonic Analysis}, vol.~42,
  no.~1, pp. 88--116, 2017.

\bibitem{Guta_2020}
M.~Gu{\c{t}}{\u{a}}, J.~Kahn, R.~Kueng, and J.~A. Tropp, ``Fast state
  tomography with optimal error bounds,'' \emph{Journal of Physics A:
  Mathematical and Theoretical}, vol.~53, no.~20, 2020, art. no.~204001.

\bibitem{lowe2022lowerboundslearningquantum}
A.~Lowe and A.~Nayak, ``Lower bounds for learning quantum states with
  single-copy measurements,'' \emph{ACM Transactions on Computation Theory},
  vol.~17, no.~1, pp. 1--42, 2025, art. no.~7.

\bibitem{10353129}
S.~Chen, B.~Huang, J.~Li, A.~Liu, and M.~Sellke, ``When does adaptivity help
  for quantum state learning?'' in \emph{2023 IEEE 64th Annual Symposium on
  Foundations of Computer Science (FOCS)}, 2023, pp. 391--404.

\bibitem{Flammia2024quantumchisquared}
S.~T. Flammia and R.~O'Donnell, ``Quantum chi-squared tomography and mutual
  information testing,'' \emph{Quantum}, vol.~8, 2024, art. no.~1381.

\bibitem{franca_et_al:LIPIcs.TQC.2021.7}
D.~S. Fran{\c{c}}a, F.~G. S.~L. Brand{\~a}o, and R.~Kueng, ``Fast and robust
  quantum state tomography from few basis measurements,'' in \emph{16th
  Conference on the Theory of Quantum Computation, Communication and
  Cryptography (TQC 2021)}, ser. Leibniz International Proceedings in
  Informatics, vol. 197, 2021, pp. 7:1--7:13.

\bibitem{PhysRevLett.105.150401}
D.~Gross, Y.-K. Liu, S.~T. Flammia, S.~Becker, and J.~Eisert, ``Quantum state
  tomography via compressed sensing,'' \emph{Physical Review Letters}, vol.
  105, no.~15, 2010, art. no.~150401.

\bibitem{liu2011universal}
Y.-K. Liu, ``Universal low-rank matrix recovery from {Pauli} measurements,'' in
  \emph{Advances in Neural Information Processing Systems}, vol.~24, 2011, pp.
  1638--1646.

\bibitem{flammia2011direct}
S.~T. Flammia and Y.-K. Liu, ``Direct fidelity estimation from few {Pauli}
  measurements,'' \emph{Physical Review Letters}, vol. 106, no.~23, 2011, art.
  no.~230501.

\bibitem{Flammia_2012}
S.~T. Flammia, D.~Gross, Y.-K. Liu, and J.~Eisert, ``Quantum tomography via
  compressed sensing: error bounds, sample complexity and efficient
  estimators,'' \emph{New Journal of Physics}, vol.~14, no.~9, 2012, art.
  no.~095022.

\bibitem{10.1214/15-AOS1382}
T.~Cai, D.~Kim, Y.~Wang, M.~Yuan, and H.~H. Zhou, ``Optimal large-scale quantum
  state tomography with {Pauli} measurements,'' \emph{The Annals of
  Statistics}, vol.~44, no.~2, pp. 682--712, 2016.

\bibitem{yu2020sampleefficienttomographypauli}
N.~Yu, ``Sample efficient tomography via {Pauli} measurements,''
  arXiv:2009.04610, 2020.

\bibitem{1614066}
D.~L. Donoho, ``Compressed sensing,'' \emph{IEEE Transactions on Information
  Theory}, vol.~52, no.~4, pp. 1289--1306, 2006.

\bibitem{NIPS2015_39461a19}
T.~Zhao, Z.~Wang, and H.~Liu, ``A nonconvex optimization framework for low rank
  matrix estimation,'' in \emph{Advances in Neural Information Processing
  Systems}, vol.~28, 2015, pp. 559--567.

\bibitem{pmlr-v49-bhojanapalli16}
S.~Bhojanapalli, A.~Kyrillidis, and S.~Sanghavi, ``Dropping convexity for
  faster semi-definite optimization,'' in \emph{29th Annual Conference on
  Learning Theory}, ser. Proceedings of Machine Learning Research, vol.~49,
  2016, pp. 530--582.

\bibitem{7536166}
R.~Sun and Z.-Q. Luo, ``Guaranteed matrix completion via non-convex
  factorization,'' \emph{IEEE Transactions on Information Theory}, vol.~62,
  no.~11, pp. 6535--6579, 2016.

\bibitem{pmlr-v48-tu16}
S.~Tu, R.~Boczar, M.~Simchowitz, M.~Soltanolkotabi, and B.~Recht, ``Low-rank
  solutions of linear matrix equations via {Procrustes} flow,'' in
  \emph{Proceedings of the 33rd International Conference on Machine Learning},
  ser. Proceedings of Machine Learning Research, vol.~48, 2016, pp. 964--973.

\bibitem{pmlr-v70-ge17a}
R.~Ge, C.~Jin, and Y.~Zheng, ``No spurious local minima in nonconvex low rank
  problems: A unified geometric analysis,'' in \emph{Proceedings of the 34th
  International Conference on Machine Learning}, ser. Proceedings of Machine
  Learning Research, vol.~70, 2017, pp. 1233--1242.

\bibitem{doi:10.1137/17M1150189}
D.~Park, A.~Kyrillidis, C.~Caramanis, and S.~Sanghavi, ``Finding low-rank
  solutions via nonconvex matrix factorization, efficiently and provably,''
  \emph{SIAM Journal on Imaging Sciences}, vol.~11, no.~4, pp. 2165--2204,
  2018.

\bibitem{kyrillidis2018provable}
A.~Kyrillidis, A.~Kalev, D.~Park, S.~Bhojanapalli, C.~Caramanis, and
  S.~Sanghavi, ``Provable compressed sensing quantum state tomography via
  non-convex methods,'' \emph{npj Quantum Information}, vol.~4, 2018, art.
  no.~36.

\bibitem{photonics10020116}
J.~L. Kim, G.~Kollias, A.~Kalev, K.~X. Wei, and A.~Kyrillidis, ``Fast quantum
  state reconstruction via accelerated non-convex programming,''
  \emph{Photonics}, vol.~10, no.~2, 2023, art. no.~116.

\bibitem{9810003}
J.~L. Kim, M.~T. Toghani, C.~A. Uribe, and A.~Kyrillidis, ``Local stochastic
  factored gradient descent for distributed quantum state tomography,''
  \emph{IEEE Control Systems Letters}, vol.~7, pp. 199--204, 2023.

\bibitem{Gao2017}
X.~Gao and L.-M. Duan, ``Efficient representation of quantum many-body states
  with deep neural networks,'' \emph{Nature Communications}, vol.~8, 2017, art.
  no.~662.

\bibitem{Torlai2018}
G.~Torlai, G.~Mazzola, J.~Carrasquilla, M.~Troyer, R.~Melko, and G.~Carleo,
  ``Neural-network quantum state tomography,'' \emph{Nature Physics}, vol.~14,
  no.~5, pp. 447--450, 2018.

\bibitem{10.21468/SciPostPhys.7.1.009}
M.~J.~S. Beach, I.~D. Vlugt, A.~Golubeva, P.~Huembeli, B.~Kulchytskyy, X.~Luo,
  R.~G. Melko, E.~Merali, and G.~Torlai, ``{QuCumber}: Wavefunction
  reconstruction with neural networks,'' \emph{SciPost Physics}, vol.~7, no.~1,
  2019, art. no.~009.

\bibitem{hsu2024quantum}
M.-C. Hsu, E.-J. Kuo, W.-H. Yu, J.-F. Cai, and M.-H. Hsieh, ``Quantum state
  tomography via nonconvex {Riemannian} gradient descent,'' \emph{Physical
  Review Letters}, vol. 132, no.~24, 2024, art. no.~240804.

\bibitem{Aaronson_2019}
S.~Aaronson, X.~Chen, E.~Hazan, S.~Kale, and A.~Nayak, ``Online learning of
  quantum states,'' \emph{Journal of Statistical Mechanics: Theory and
  Experiment}, vol. 2019, no.~12, 2019, art. no.~124019.

\bibitem{chen2020practicaladaptivealgorithmsonline}
Y.~Chen and X.~Wang, ``More practical and adaptive algorithms for online
  quantum state learning,'' arXiv:2006.01013, 2020.

\bibitem{yang2020revisiting}
F.~Yang, J.~Jiang, J.~Zhang, and X.~Sun, ``Revisiting online quantum state
  learning,'' in \emph{Proceedings of the 34th AAAI Conference on Artificial
  Intelligence}, vol.~34, 2020, pp. 6607--6614.

\bibitem{Chen2024adaptiveonline}
X.~Chen, E.~Hazan, T.~Li, Z.~Lu, X.~Wang, and R.~Yang, ``Adaptive online
  learning of quantum states,'' \emph{Quantum}, vol.~8, 2024, art. no.~1471.

\bibitem{tseng2024onlinelearningquantumstates}
W.-F. Tseng, K.-C. Chen, Z.-H. Xiao, and Y.-H. Li, ``Online learning of quantum
  states with logarithmic loss via {VB-FTRL},'' in \emph{Proceedings of the
  36th International Conference on Algorithmic Learning Theory}, ser.
  Proceedings of Machine Learning Research, vol. 272, 2025, pp. 1285--1312.

\bibitem{lumbreras2024learningpurequantumstates}
J.~Lumbreras, M.~Terekhov, and M.~Tomamichel, ``Learning pure quantum states
  almost without regret,'' \emph{npj Quantum Information}, 2026, advance online
  publication.

\bibitem{Fawzi2024}
O.~Fawzi, R.~Kueng, D.~Markham, and A.~Oufkir, ``Learning properties of quantum
  states without the {IID} assumption,'' \emph{Nature Communications}, vol.~15,
  2024, art. no.~9677.

\bibitem{meyer2025onlinelearningpurestates}
M.~Meyer, S.~Adhikary, N.~Guo, and P.~Rebentrost, ``Online learning of pure
  states is as hard as mixed states,'' in \emph{Advances in Neural Information
  Processing Systems}, vol.~38, 2025, pp. 149\,425--149\,455.

\bibitem{Anshu2024}
A.~Anshu and S.~Arunachalam, ``A survey on the complexity of learning quantum
  states,'' \emph{Nature Reviews Physics}, vol.~6, no.~1, pp. 59--69, 2024.

\bibitem{doi:10.1137/18M120275X}
S.~Aaronson, ``Shadow tomography of quantum states,'' \emph{SIAM Journal on
  Computing}, vol.~49, no.~5, pp.
  \mbox{STOC18\symbol{45}368}--\mbox{STOC18\symbol{45}394}, 2020.

\bibitem{Huang2020}
H.-Y. Huang, R.~Kueng, and J.~Preskill, ``Predicting many properties of a
  quantum system from very few measurements,'' \emph{Nature Physics}, vol.~16,
  no.~10, pp. 1050--1057, 2020.

\bibitem{doi:10.1126/science.abk3333}
H.-Y. Huang, R.~Kueng, G.~Torlai, V.~V. Albert, and J.~Preskill, ``Provably
  efficient machine learning for quantum many-body problems,'' \emph{Science},
  vol. 377, no. 6613, 2022, art. no.~eabk3333.

\bibitem{Youssry_2019}
A.~Youssry, C.~Ferrie, and M.~Tomamichel, ``Efficient online quantum state
  estimation using a matrix-exponentiated gradient method,'' \emph{New Journal
  of Physics}, vol.~21, no.~3, 2019, art. no.~033006.

\bibitem{Zhang2020}
K.~Zhang, S.~Cong, K.~Li, and T.~Wang, ``An online optimization algorithm for
  the real-time quantum state tomography,'' \emph{Quantum Information
  Processing}, vol.~19, no.~10, 2020, art. no.~361.

\bibitem{pmlr-v238-tsai24a}
C.-E. Tsai, H.-C. Cheng, and Y.-H. Li, ``Fast minimization of expected
  logarithmic loss via stochastic dual averaging,'' in \emph{Proceedings of the
  27th International Conference on Artificial Intelligence and Statistics},
  ser. Proceedings of Machine Learning Research, vol. 238, 2024, pp.
  2908--2916.

\bibitem{lin2021maximumlikelihoodquantumstatetomography}
C.-M. Lin, H.-C. Cheng, and Y.-H. Li, ``Maximum-likelihood quantum state
  tomography by {Cover}'s method with non-asymptotic analysis,''
  arXiv:2110.00747, 2021.

\bibitem{lin2022maximumlikelihoodquantumstatetomography}
C.-M. Lin, Y.-M. Hsu, and Y.-H. Li, ``Maximum-likelihood quantum state
  tomography by {Soft-Bayes},'' arXiv:2012.15498v3, 2022.

\bibitem{tsai2022fasterstochasticfirstordermethod}
C.-E. Tsai, H.-C. Cheng, and Y.-H. Li, ``Faster stochastic first-order method
  for maximum-likelihood quantum state tomography,'' arXiv:2211.12880, 2022.

\bibitem{rambach2023efficient}
M.~Rambach, A.~Youssry, M.~Tomamichel, and J.~Romero, ``Efficient quantum state
  tracking in noisy environments,'' \emph{Quantum Science and Technology},
  vol.~8, no.~1, 2023, art. no.~015010.

\bibitem{hu2009accelerated}
C.~Hu, J.~T. Kwok, and W.~Pan, ``Accelerated gradient methods for stochastic
  optimization and online learning,'' in \emph{Advances in Neural Information
  Processing Systems}, vol.~22, 2009, pp. 781--789.

\bibitem{bottou2010large}
L.~Bottou, ``Large-scale machine learning with stochastic gradient descent,''
  in \emph{Proceedings of COMPSTAT'2010}, 2010, pp. 177--186.

\bibitem{duchi2011adaptive}
J.~Duchi, E.~Hazan, and Y.~Singer, ``Adaptive subgradient methods for online
  learning and stochastic optimization,'' \emph{Journal of Machine Learning
  Research}, vol.~12, no.~61, pp. 2121--2159, 2011.

\bibitem{li2018statistical}
T.~Li, L.~Liu, A.~Kyrillidis, and C.~Caramanis, ``Statistical inference using
  {SGD},'' in \emph{Proceedings of the 32nd AAAI Conference on Artificial
  Intelligence}, vol.~32, 2018, pp. 3571--3578.

\bibitem{jin2016provable}
C.~Jin, S.~M. Kakade, and P.~Netrapalli, ``Provable efficient online matrix
  completion via non-convex stochastic gradient descent,'' in \emph{Advances in
  Neural Information Processing Systems}, vol.~29, 2016, pp. 4527--4535.

\bibitem{recht2013parallel}
B.~Recht and C.~R{\'e}, ``Parallel stochastic gradient algorithms for
  large-scale matrix completion,'' \emph{Mathematical Programming Computation},
  vol.~5, no.~2, pp. 201--226, 2013.

\bibitem{gemulla2011large}
R.~Gemulla, E.~Nijkamp, P.~J. Haas, and Y.~Sismanis, ``Large-scale matrix
  factorization with distributed stochastic gradient descent,'' in
  \emph{Proceedings of the 17th ACM SIGKDD International Conference on
  Knowledge Discovery and Data Mining}, 2011, pp. 69--77.

\bibitem{de2015global}
C.~De~Sa, C.~R{\'e}, and K.~Olukotun, ``Global convergence of stochastic
  gradient descent for some non-convex matrix problems,'' in \emph{Proceedings
  of the 32nd International Conference on Machine Learning}, ser. Proceedings
  of Machine Learning Research, vol.~37, 2015, pp. 2332--2341.

\bibitem{ge2015escaping}
R.~Ge, F.~Huang, C.~Jin, and Y.~Yuan, ``Escaping from saddle points---online
  stochastic gradient for tensor decomposition,'' in \emph{Proceedings of the
  28th Conference on Learning Theory}, ser. Proceedings of Machine Learning
  Research, vol.~40, 2015, pp. 797--842.

\bibitem{pelecanos2025mixedstatetomographyreduces}
A.~Pelecanos, J.~Spilecki, E.~Tang, and J.~Wright, ``Mixed state tomography
  reduces to pure state tomography,'' in \emph{2026 IEEE 67th Annual Symposium
  on Foundations of Computer Science (FOCS)}, 2026, arXiv:2511.15806.

\bibitem{cohen2016geometric}
M.~B. Cohen, Y.~T. Lee, G.~L. Miller, J.~W. Pachocki, and A.~Sidford,
  ``Geometric median in nearly linear time,'' in \emph{Proceedings of the 48th
  Annual ACM Symposium on Theory of Computing}, 2016, pp. 9--21.

\bibitem{jain2016streaming}
P.~Jain, C.~Jin, S.~M. Kakade, P.~Netrapalli, and A.~Sidford, ``Streaming
  {PCA}: Matching matrix {Bernstein} and near-optimal finite sample guarantees
  for {Oja}'s algorithm,'' in \emph{29th Annual Conference on Learning Theory},
  ser. Proceedings of Machine Learning Research, vol.~49, 2016, pp. 1147--1164.

\bibitem{zhu2013angles}
P.~Zhu and A.~V. Knyazev, ``Angles between subspaces and their tangents,''
  \emph{Journal of Numerical Mathematics}, vol.~21, no.~4, pp. 325--340, 2013.

\end{thebibliography}
}

\end{document}